\documentclass[a4paper,10pt]{article}
\usepackage[utf8]{inputenc}
\usepackage{amsthm}
\usepackage{amsmath}
\usepackage{authblk}
\usepackage{cite}
\usepackage{amssymb}
\usepackage{hyperref}
\newtheorem{theorem}{Theorem}

\newtheorem{remark}{Remark}
\newtheorem{lemma}{Lemma}
\newtheorem{corollary}{Corollary}
\newtheorem{proposition}{Proposition}

\title{A Generalization of the Borkar-Meyn Theorem for Stochastic Recursive Inclusions}
\author[1]{
Arunselvan Ramaswamy
}
\author[2]{
Shalabh Bhatnagar
}
\affil[1]{\texttt{arunselvan@csa.iisc.ernet.in}}
\affil[2]{\texttt{shalabh@csa.iisc.ernet.in}}
\affil[1,2]{Department of Computer Science and Automation,
Indian Institute of Science,
Bangalore - 560012, India.}
\begin{document}

\maketitle

\begin{abstract}
In this paper the stability theorem of \textbf{Borkar and Meyn} is extended to
 include the case when the mean field is a set-valued map. Two different sets of
 sufficient conditions are presented that guarantee the `stability and 
 convergence' of stochastic recursive inclusions. 
Our work builds on the works
of \textbf{Bena\"{i}m, Hofbauer and Sorin} as well as \textbf{Borkar and Meyn}. 
As a corollary to one of the main theorems, a natural
generalization of the \textit{Borkar and Meyn Theorem} follows.
In addition, the original theorem of \textit{Borkar and Meyn} 
is shown to hold under slightly relaxed assumptions.
As an application to one of the main theorems we discuss 
a solution to the `approximate drift problem'. Finally, we analyze the stochastic
gradient algorithm with ``constant error gradient estimators''
as yet another application of our main result.
\end{abstract}
\section{Introduction}
\label{introduction}
\paragraph{}
Consider the following recursion in $\mathbb{R}^{d}$ ($d \ge 1$):
\begin{equation} \label{eq:basicrecursion}
 x_{n+1}\ =\ x_{n}\ +\ a(n)\ \left[ h(x_{n})\ +\ M_{n+1} \right],\ for\ n \ge 0,\ \text{where}
\end{equation}
\begin{itemize}
 \item[(i)] $h: \mathbb{R}^d \to \mathbb{R}^d$ is a Lipschitz continuous function.
 \item[(ii)] $a(n) > 0$, for all $n$, is the step-size sequence satisfying
 $\sum_{n=0}^\infty a(n) = \infty$ and $\sum_{n=0}^\infty a(n)^2 < \infty$.
 \item[(iii)] $M_n$, $n \ge 1$, is a sequence of martingale difference terms that constitute 
 the noise.
\end{itemize}

The stochastic recursion given by (\ref{eq:basicrecursion})
is often referred to as a \textit{stochastic recursive equation} \textit{(SRE)}.
A powerful method to analyze the limiting behavior of
(\ref{eq:basicrecursion}) is the \textit{ODE
(Ordinary Differential Equation) method}. Here the limiting behavior
of the algorithm is described in terms of the asymptotics of the solution to the \textit{ODE}
\begin{equation}\nonumber
 \dot{x}(t) = h(x(t)).
\end{equation}
This method was introduced by \textbf{Ljung} \cite{Ljung77} in $1977$. For
a detailed exposition on the subject and a survey of results,
the reader is referred to \textbf{Kushner and Yin} \cite{KushnerYin}
as well as \textbf{Borkar} \cite{BorkarBook}.

\paragraph{}
In $1996$,
\textbf{Bena\"{i}m} \cite{Benaim96} showed that the asymptotic behavior of 
a stochastic recursive equation can be studied by analyzing the
asymptotic behavior of the \textit{associated o.d.e.} However no assumptions were made 
on the dynamics of the \textit{o.d.e.} Specifically, he developed
sufficient conditions which guarantee that limit sets of the continuously interpolated stochastic iterates 
are compact, connected, internally chain transitive and invariant
sets of the \textit{associated o.d.e.} The results found in \cite{Benaim96} are 
generalized in \cite{Benaim99}; further studies were made by \textbf{Bena\"{i}m and Hirsch}
in \cite{BenaimHirsch}.
The assumptions made in \cite{Benaim96} are sometimes referred to as the `classical assumptions'.
One of the key assumptions used by Bena\"{i}m
to prove this convergence theorem is the \textit{almost sure} boundedness of the iterates
\textit{i.e.,} \textit{stability} of the iterates.
In $1999$, \textbf{Borkar and Meyn}
\cite{Borkar99} developed sufficient conditions which guarantee both the stability and convergence
of stochastic recursive equations. These assumptions were consistent with those developed in 
\cite{Benaim96}.
In this paper we refer to the main result of \textbf{Borkar and Meyn} 
colloquially as the \textit{Borkar-Meyn Theorem}.
In the same paper \cite{Borkar99}, several
applications to problems from reinforcement learning have also been discussed.
Another set of sufficient conditions for $SRE$'s were developed by \textbf{Andrieu, Moulines and Priouret}
\cite{Andrieu} using global Lyapunov functions that guarantee the stability and convergence
of the iterates.


\paragraph{}
In $2005$, \textbf{Bena\"{i}m, Hofbauer and Sorin} \cite{Benaim05} showed that the dynamical
systems approach can be extended to the situation where
the mean fields are \textit{set-valued}.
The algorithms considered were of the form:
\begin{equation}\label{eq:recursiveinclusion}
  x_{n+1}\ = x_{n}\ +\ a(n) \left[ y_{n}\ +\ M_{n+1} \right],\ for\ n \ge 0,\ \text{where} 
\end{equation}
 
\begin{itemize}
 \item[(i)] $y_n \in h(x_n)$ and $h: \mathbb{R}^d \to \{subsets\ of\ \mathbb{R}^d\}$ is a Marchaud map.
 For the definition of Marchaud maps the reader is referred to section
 ~\ref{definitions}.
 \item[(ii)] $a(n) > 0$, for all $n \ge 0$, is the step-size sequence satisfying
 $\sum_{n=0}^\infty a(n) = \infty$ and $\sum_{n=0}^\infty a(n)^2 < \infty$.
 \item[(iii)] $M_n$, $n \ge 1$, is a sequence of martingale difference terms.
\end{itemize}
            \paragraph{}
            A recursion such as 
(\ref{eq:recursiveinclusion}) is also called \textit{stochastic recursive inclusion} \textit{(SRI)}.
Since
a differential equation can be seen as a special case of a differential
inclusion wherein $h(x)$ is a cardinality one set for all $ \ x \in \mathbb{R}^{d}$,
\textit{SRE} (\ref{eq:basicrecursion}) can be seen as a special case of
\textit{SRI} (\ref{eq:recursiveinclusion}). 
\paragraph{}
The main aim
of this paper is to extend the original \textit{Borkar-Meyn theorem} 
to the case of stochastic recursive inclusions. 
We present two overlapping yet different
sets of assumptions, in Sections~\ref{assumptions} and \ref{AccuMainSec} 
respectively, that guarantee
the stability and convergence of a \textit{SRI} given by (\ref{eq:recursiveinclusion}).
As a consequence of our main results, Theorems \ref{main} and \ref{AccuMain}, we
present a couple of interesting extensions to the original theorem of Borkar and Meyn
in Section~\ref{GenBorkarMeynSec}. Using the frameworks presented herein
we provide a solution to the problem of approximate drift in Section~\ref{ApproximateDrift}.
For more details on the approximate drift problem the reader is referred to Borkar \cite{BorkarBook}. 
In Section~\ref{finaldisc} we discuss the generality, ease of verifiability and
we also try to explain why the assumptions are ``natural'' in some sense.
\paragraph{} 
\textit{Stochastic gradient descent} ($SGD$) is an important method to find minima of (continuously) 
differentiable functions. When implementing the corresponding approximation algorithm
(See (\ref{sgdi}) in Section~\ref{sgdsec}) using gradient estimators,
an error is made at each step in calculating the gradient of the objective function.
Lets call this error the ``approximation error''.
This is the case when using gradient estimators such as
Kiefer-Wolfowitz, simultaneous perturbation
stochastic approximation (SPSA) and smoothed functional (SF) schemes, see \cite{shalabh}.
Suppose the \textit{perturbation parameters} of the aforementioned estimators
are kept \textit{constant}, then the ``approximation error'' is bounded by a constant that
depends on the size of the perturbation parameters.
We call such estimators \textit{constant-error gradient estimators}.
In Section~\ref{sgdsec} we analyze the 
stochastic gradient approximation algorithm that uses a constant-error gradient
estimator.
Using Theorem~\ref{AccuMain} we show that the iterates are stable
and converge to a $\delta$-neighborhood of the minimum set, for a specified $\delta(>0)$. 
Essentially, our framework gives
a threshold $\epsilon(\delta)$ for the ``approximation error'' so that the 
stochastic gradient approximation algorithm is stable and converges to
a $\delta$-neighborhood of the minimum set.
\paragraph{}
\textit{It is worth noting that prior to this paper
one could only claim that an $SGD$ using constant-error gradient
estimators will only converge to some neighborhood of the minimum set with high probability.
On the other hand, our framework
guarantees \textbf{almost sure} convergence to a small neighborhood of the
minimum set.}
\section{Preliminaries and Assumptions}

\subsection{Definitions and Notations} \label{definitions}
\paragraph{}
The definitions and notations used in this paper are similar to
those in Bena\"{i}m et. al. \cite{Benaim05},
Aubin et. al. \cite{Aubin}, \cite{AubinSet} and Borkar \cite{BorkarBook}. In this section, we 
present a few for easy reference.
\paragraph{}
A set-valued map $h: \mathbb{R}^n \to \{subsets\ of\ \mathbb{R}^m$ \} 
is called a \textit{Marchaud map} if it satisfies
the following properties:
\begin{itemize}
 \item[(i)] For each $x$ $\in \mathbb{R}^{n}$, $h(x)$ is convex and compact.
 \item[(ii)] \textit{(point-wise boundedness)} For each $x \in \mathbb{R}^{n}$,  
 $\underset{w \in h(x)}{\sup}$ $\lVert w \rVert$
 $< K \left( 1 + \lVert x \rVert \right)$ for some $K > 0$.
 \item[(iii)] $h$ is an \textit{upper-semicontinuous} map. 
 We say that $h$ is upper-semicontinuous,
  if given sequences $\{ x_{n} \}_{n \ge 1}$ (in $\mathbb{R}^{n}$) and 
  $\{ y_{n} \}_{n \ge 1}$ (in $\mathbb{R}^{m}$)  with
  $x_{n} \to x$, $y_{n} \to y$ and $y_{n} \in h(x_{n})$, $n \ge 1$, 
  implies that $y \in h(x)$. In other words
  the graph of $h$, $\left\{ (x, y) \ : \ y \in h(x),\ x\in \mathbb{R}^{n} \right\}$,
  is closed in $\mathbb{R}^{n} \times \mathbb{R}^{m}$.
\end{itemize}
Let $H$ be a Marchaud map on $\mathbb{R}^d$.
The differential inclusion (DI) given by
\begin{equation} \label{di}
\dot{x} \ \in \ H(x)
\end{equation}
is guaranteed to have at least one solution that is absolutely continuous. 
The reader is referred to \cite{Aubin} for more details.
We say that $\textbf{x} \in \sum$ if $\textbf{x}$ 
is an absolutely continuous map that satisfies (\ref{di}).
The \textit{set-valued semiflow}
$\Phi$ associated with (\ref{di}) is defined on $[0, + \infty) \times \mathbb{R}^d$ as: \\
$\Phi_t(x) = \{\textbf{x}(t) \ | \ \textbf{x} \in \sum , \textbf{x}(0) = x \}$. Let
$B \times M \subset [0, + \infty) \times \mathbb{R}^k$ and define
\begin{equation}\nonumber
 \Phi_B(M) = \underset{t\in B,\ x \in M}{\bigcup} \Phi_t (x).
\end{equation}
\indent
Let $M \subseteq \mathbb{R}^d$, the $\omega-limit$ \textit{set} be defined by
$
 \omega_{\Phi}(M) = \bigcap_{t \ge 0} \ \overline{\Phi_{[t, +\infty)}(M)}.
$
Similarly the \textit{limit set} of a solution $\textbf{x}$ is given by
$L(x) = \bigcap_{t \ge 0} \ \overline{\textbf{x}([t, +\infty))}$.
\\ \indent
$M \subseteq \mathbb{R}^d$ is \textit{invariant} if for every $x \in M$ there exists 
a trajectory, $\textbf{x}$, entirely in $M$
with $\textbf{x}(0) = x$. In other words, $\textbf{x} \in \sum$  with $\textbf{x}(t) \in M$,
for all $t \ge 0$.
\\ \indent
Let $x \in \mathbb{R}^d$ and $A \subseteq \mathbb{R}^d$, then
$d(x, A) : = \inf \{\lVert a- y \rVert \ | \ y \in A\}$. We define the $\delta$-\textit{open neighborhood}
of $A$ by $N^\delta (A) := \{x \ |\ d(x,A) < \delta \}$. The 
$\delta$-\textit{closed neighborhood} of $A$ 
is defined by $\overline{N^\delta} (A) := \{x \ |\ d(x,A) \le \delta \}$.
The open ball of radius $r$ around the origin is represented by $B_r(0)$,
while the closed ball is represented by $\overline{B}_r(0)$.
\\ \indent
\textit{Internally Chain Transitive Set}:
$M \subset \mathbb{R}^{d}$ is said to be
internally chain transitive if $M$ is compact and for every $x, y \in M$,
$\epsilon >0$ and $T > 0$ we have the following: There exist $\Phi^{1}, \ldots, \Phi^{n}$ that
are $n$ solutions to the differential inclusion $\dot{x}(t) \in h(x(t))$,
a sequence $x_1(=x), \ldots, x_{n+1} (=y) \subset M$
and $n$ real numbers 
$t_{1}, t_{2}, \ldots, t_{n}$ greater than $T$ such that: $\Phi^i_{t_{i}}(x_i) \in N^\epsilon(x_{i+1})$ and
$\Phi^{i}_{[0, t_{i}]}(x_i) \subset M$ for $1 \le i \le n$. The sequence $(x_{1}(=x), \ldots, x_{n+1}(=y))$
is called an $(\epsilon, T)$ chain in $M$ from $x$ to $y$.
\\ \indent
$A \subseteq \mathbb{R}^d$ is an \textit{attracting set} if it is compact
and there exists a neighborhood $U$ such that for any $\epsilon > 0$,
$\exists \ T(\epsilon) \ge 0$ with $\Phi_{[T(\epsilon), +\infty)}(U) \subset
N^{\epsilon}(A)$. Such a $U$ is called the \textit{fundamental neighborhood} of $A$. 
In addition to being compact if the \textit{attracting set} is also invariant then
it is called an \textit{attractor}.
The \textit{basin
of attraction } of $A$ is given by $B(A) = \{x \ | \ \omega_\Phi(x) \subset A\}$. It is called
\textit{Lyapunov stable} if for all $\delta > 0$, $\exists \ \epsilon > 0$ such that
$\Phi_{[0, +\infty)}(N^\epsilon(A)) \subseteq N^\delta(A)$.
We use $T(\epsilon)$ and $T_\epsilon$ interchangeably to denote the dependence of $T$ on $\epsilon$.
\\ \indent 
We define the \textit{lower and upper limits} of sequences of sets.
Let $\{K_{n}\}_{n \ge 1}$ be a sequence of sets in $\mathbb{R}^{d}$. 
\begin{enumerate}
 \item The \textit{lower limit} of $\{K_{n} \}_{n \ge 1}$
is given by, 
$Liminf_{n \to \infty} K_{n}$ $:=$ $\{x \ |\ 
 \underset{n \to \infty}{\lim} d(x, K_{n}) = 0 \}$.
 \item The \textit{upper-limit} of $\{K_{n}\}_{n \ge 1}$ is
given by, 
$Limsup_{n \to \infty} K_n$ $:= \ \{y \ | \ 
\underset{n \to \infty}{\underline{lim}}d(y, K_n)= 0 \}$. \\
We may interpret that the lower-limit collects the limit points of  
$\{K_n\}_{n \ge 1}$ while the upper-limit
collects its accumulation points.
\end{enumerate}
\subsection{The assumptions} \label{assumptions}
\paragraph{}
Recall that we have the following recursion in $\mathbb{R}^{d}$:
\begin{equation} \nonumber
 x_{n+1} = x_{n} + a(n) \left[ y_{n} + M_{n+1} \right], \text{ where } y_n \in h(x_n).
\end{equation}
We state our assumptions below:
\begin{itemize}
 \item[\textbf{(A1)}] $h: \mathbb{R}^d \to \{\text{subsets of }\mathbb{R}^d\}$ is a Marchaud map.
 \item[\textbf{(A2)}] $\{ a(n) \}_{n \ge 0}$ is a scalar sequence such that: $a(n) > 0$ $\forall n$, 
 $\underset{n \ge 0}{\sum} a(n) = \infty$
 and 
 $\underset{n \ge 0}{\sum} a(n)^{2} < \infty$. Without loss of generality we let
 $\underset{n}{sup}\ a(n) \le 1$.  
 \item[\textbf{(A3)}] 
 \begin{itemize}
  \item[]$\{ M_{n}\}_{n \ge 1}$ is a martingale difference sequence with respect to
 the filtration \\
 $\mathcal{F}_{n}$ $:=$ $ \sigma \left( x_{0}, M_{1}, \ldots, M_{n} \right) $, 
$n \ge 0$.
\item[(i)] $\{ M_{n}\}_{n \ge 1}$ is a square integrable sequence.
\item[(ii)] $E[\lVert M_{n+1} \rVert ^{2} | \mathcal{F}_{n}]$ $\le$ $K \left( 1 + \lVert x_{n} \rVert ^{2} \right)$, for 
$n \ge 0$ and some
 constant $K > 0$. Without loss of generality assume that the same constant, $K$,
 works for both the point-wise boundedness condition of $(A1)$
 (see condition (ii) in the definition of Marchaud map in Section \ref{definitions}) and $(A3)$.
 \end{itemize}

 \end{itemize}
$\\ $ \indent
 For $c \ge 1$ and $x \in \mathbb{R}^{d}$, define $h_{c}(x) = \{y \ |\ cy \in h(cx) \}$. Further,
for each $x \in \mathbb{R}^d$, define
$h_{\infty}(x) := \overline{\ Liminf_{c \rightarrow \infty } \ h_{c}(x)}$ \textit{i.e.}
the closure of the \textit{lower-limit} of $\{h_c(x)\}_{c \ge 1}$.
$\\ $
 \begin{itemize}
 \item[\textbf{(A4)}] $h_{\infty}(x)$ is non-empty for all $x \in \mathbb{R}^{d}$.
 Further, the differential inclusion $\dot {x}(t) \in h_{\infty}(x(t))$ has an 
 attracting set, $\mathcal{A}$, with $\overline{B}_{1}(0)$ as a subset of its
 fundamental neighborhood. This attracting set is such that $\mathcal{A} \subseteq B_{1}(0)$.
\item[\textbf{(A5)}] Let $c_{n} \ge 1$ be an increasing sequence of integers such that
$c_{n} \uparrow \infty$ as $n \to \infty$. Further, let
$x_{n} \ \rightarrow \ x$ and $y_{n} \ \rightarrow \ y$ as 
$n \ \rightarrow \infty$, such that $y_{n} \in h_{c_{n}}(x_{n})$, $\forall n$, 
then $y \in h_{\infty}(x)$.
\end{itemize}

\paragraph{}
Since the attracting set, $\mathcal{A} \subseteq B_1(0)$, is compact we conclude that
$\underset{x \in \mathcal{A}}{\sup} \lVert x \rVert < 1$. To see this, for all 
$x \in \mathcal{A}$ define $\delta(x) := \underset{y \in \overline{B}_{\epsilon(x)}(x)}{\sup} \lVert y \rVert$,
where $\epsilon(x) > 0$ and $\overline{B}_{\epsilon(x)}(x) \subseteq B_{1}(0)$. For all 
$x \in \mathcal{A}$ we have $\delta(x) < 1$. Further, $\{B_{\epsilon(x)}(x) \ | \ x \in \mathcal{A}\}$
is an open cover of $\mathcal{A}$. Let $\{B_{\epsilon(x_i)}(x_i) \ | \ 1 \le i \le n \}$
be a finite sub-cover and $\delta := \underset{1 \le i \le n}{\max} \delta(x_i)$. Clearly,
it follows that $\underset{x \in \mathcal{A}}{\sup} \lVert x \rVert \le \delta < 1$.
Define $\delta_1 := \underset{x \in \mathcal{A}}{\sup} \lVert x \rVert$ and pick real numbers
$\delta_2,\ \delta_3$ and $\delta_4$ such that $\underset{x \in \mathcal{A}}{\sup} \lVert x \rVert = \delta_1 < \delta_2 < \delta_3 < \delta_4 <1$.
We shall use this sequence later on.
\paragraph{}
Assumptions $(A1) - (A3)$ are the same as in Bena\"{i}m \cite{Benaim05}. However, the assumption
on the stability of the iterates is replaced by $(A4)$ and $(A5)$.
We show that $(A4)$ and $(A5)$ are sufficient conditions to ensure stability of iterates.
We start by observing that $h_c$ and $h_\infty$ are Marchaud maps, where $c \ge 1$.
Further, we show that the constant associated with the point-wise
boundedness property is $K$ of $(A1)$ and $(A3)$.
\begin{proposition}\label{propo}
 $h_{\infty}$ and $h_{c}$, $c \ge 1$, are Marchaud maps.
\end{proposition}
\begin{proof}
Fix $c \ge 1$ and $x \in \mathbb{R}^{d}$. To prove that $h_{c}(x)$ is compact, we show
that it is closed and bounded.
For $n \ge 1$, let $y_{n} \in h_{c}(x)$  and let
$\underset{n \to \infty}{\lim} y_{n} = y$. It follows that $c y_{n} \in h(cx)$ for each $n \ge 1$
and $\underset{n \to \infty}{\lim} cy_{n} = cy$. Since $h(cx)$ is closed,
we have that $cy \in h(cx)$ and 
$y \in h_{c}(x)$.
If we show that $h_c$ is point-wise bounded then we can conclude that $h_c(x)$
is compact. To prove the aforementioned,
let $y \in h_{c}(x)$, then $cy \in h(cx)$. Since $h$ satisfies $(A1)(ii)$, we have that
\begin{equation}\nonumber
 c \lVert y \rVert \le \ K \left( 1+ \lVert cx \rVert \right),\ \text{hence}
 \end{equation}
\begin{equation}\nonumber
 \lVert y \rVert \le \ K \left( \frac{1}{c}+ \lVert x \rVert \right).
\end{equation}
Since $c (\ge 1)$ and $x$ is arbitrarily chosen, $h_{c}$ is point-wise bounded and
the compactness of $h_c(x)$ follows.
The set $h_c(x) = \{ z/c \ |\ z \in h(cx)\}$ is convex since $h(cx)$ is convex
and $h_c(x)$ is obtained by scaling it by $\frac{1}{c}$.
\paragraph{}
Next, we show that $h_{c}(x)$ is upper-semicontinuous. 
Let $\underset{n \to \infty}{\lim}$ $x_{n} = x$, $\underset{n \to \infty}{\lim}$
$y_{n}  = y$ and $y_{n} \in h_c (x_{n})$, $\forall$ $n \ge 1$. 
We need to show that $y \in h_c (x)$.
We have that $c y_{n} \in h(cx_{n})$ for each $n \ge 1$. Since $\underset{n \to \infty}{\lim}$ $cx_{n} = cx$
and $\underset{n \to \infty}{\lim}$ $cy_{n} = cy$, we conclude that $cy \in h(cx)$ since $h$
is assumed to be upper-semicontinuous.
\paragraph{}
It is left to show that $h_{\infty}(x)$, $x \in \mathbb{R}^{d}$
is a Marchaud map. To prove that 
$\lVert z \rVert$ $\le K \left( 1 + \lVert x \rVert \right)$ for all $z \in h_{\infty}(x)$,
it is enough to prove that
$\lVert y \rVert$ $\le K \left( 1 + \lVert x \rVert \right)$ for all 
$y \in Liminf_{c \rightarrow \infty } \ h_{c}(x)$. Fix some
$y \in Liminf_{c \rightarrow \infty } \ h_{c}(x)$ then there exist
$z_{n} \in h_{n}(x)$, $n \ge 1$, such that $\underset{n \to \infty}{\lim}$ $\lVert y - z_{n} \rVert = 0$.
We have that 
\begin{center}
$\lVert y \rVert$ $\le$ $\lVert y - z_{n} \rVert$ + $\lVert z_{n} \rVert$. 
\end{center}
Since  $h_{c}$, $c \ge 1$, is point-wise bounded (the constant associated is independent of $c$
and equals $K$), the above inequality becomes
\begin{center}
$\lVert y \rVert$ $\le$ $\lVert y - z_{n} \rVert$ + $K \left( 1 + \lVert x \rVert \right)$.
\end{center}
Letting $n \to \infty$ in the above inequality, we obtain 
$\lVert y \rVert$ $\le$ $K \left( 1 + \lVert x \rVert \right)$.
Recall that 
$h_{\infty}(x) = \overline{\ Liminf_{c \rightarrow \infty } \ h_{c}(x)}$, 
hence it is compact.
\paragraph{}
Again, to show that $h_{\infty}(x)$ is convex, for each $x \in \mathbb{R}^{d}$, 
we start by proving that $Liminf_{c \rightarrow \infty } \ h_{c}(x)$ is convex.
Let $u, v \in Liminf_{c \rightarrow \infty } \ h_{c}(x)$ and $0 \le t \le 1$. 
We need to show that $tu + (1-t)v \in 
Liminf_{c \rightarrow \infty } \ h_{c}(x)$. 
Consider an arbitrary sequence $\{c_{n}\}_{n\ge 1}$
such that $c_{n} \to \infty$, then there exist $u_{n}, v_{n} \in h_{c_{n}}(x)$
such that $\lVert u_{n} - u \rVert$ and $\lVert v_{n} - v \rVert$ $\to 0$ as 
$c_{n} \to \infty$. Since $h_{c_{n}}(x)$ is convex, it follows that
$tu_{n} + (1-t)v_{n} \in h_{c_{n}}(x)$, further
\begin{equation}\nonumber
 \underset{c_{n} \to \infty}{\lim} \left( \ tu_{n} + (1-t)v_{n} \right) 
 \ =\ tu + (1-t)v.
\end{equation}
Since we started with an arbitrary sequence $c_{n} \to \infty$, it follows that
$tu + (1-t)v \in Liminf_{c \rightarrow \infty } \ h_{c}(x)$. 
Now we can prove that $h_\infty(x)$ is convex.
Let $u, v \in h_{\infty}(x)$. Then
$\exists$ $\{u_n\}_{n \ge 1}$ and $\{v_n\}_{n \ge 1} \subseteq 
Liminf_{c \rightarrow \infty } \ h_{c}(x)$ 
such that $u_n \to u$
and $v_n \to v$ as $n \to \infty$. We need to show that $tu + (1-t)v \in h_{\infty}(x)$,
for $0 \le t \le 1$.
Since $tu_n + (1-t)v_n \in Liminf_{c \rightarrow \infty } \ h_{c}(x)$,
the desired result is obtained by letting $n \to \infty$ in $tu_n + (1-t)v_n$.
\paragraph{}
Finally, we show that $h_{\infty}$ is upper-semicontinuous. Let
$\underset{n \to \infty}{\lim}$ $x_{n}=x$, $\underset{n \to \infty}{\lim}$ $y_{n}=y$
and $y_{n} \in h_{\infty}(x_{n})$, $\forall$ $n \ge 1$. We need to show that 
$y \in h_{\infty}(x)$. 
Since $y_{n} \in h_{\infty}(x_{n})$, $\exists$ 
$z_n \in Liminf_{c \rightarrow \infty } \ h_{c}(x_n)$ such that 
$\lVert z_n - y_n \rVert < \frac{1}{n}$.
Since $z_{n} \in Liminf_{c \rightarrow \infty } \ h_{c}(x_n)$, $n \ge 1$, 
it follows that there exist
$c_{n}$ such that for all $c \ge c_{n}$, $d \left(z_{n}, h_{c}(x_{n}) \right) < \frac{1}{n}$.
In particular, $\exists\ u_{n}$ $\in h_{c_{n}}(x_{n})$ such that
$\lVert z_{n} - u_{n} \rVert$ $< \frac{1}{n}$. 
We choose the sequence $\{c_{n}\}_{n \ge 1}$ such that $c_{n+1} > c_{n}$ for each $n \ge 1$.
We now have the following:
$\underset{n \to \infty}{\lim}$ $u_{n} = y$, $u_{n} \in h_{c_{n}}(x_{n})$ 
$\forall$ $n$ and
$\underset{n \to \infty}{\lim}$ $x_{n} = x$. It follows directly from assumption $(A5)$
that $y \in h_{\infty}(x)$.
\end{proof}
\newpage
\section{Stability and convergence of stochastic recursive inclusions} \label{sec3}
\paragraph{}
We begin by providing a brief outline of our approach to prove the stability
of a $SRI$ under assumptions $(A1)-(A5)$.
First we divide the time line, $[0, \infty)$, approximately
into intervals of length $T$. 
We shall explain later how we choose and fix $T$.
Then we construct a \textit{linearly interpolated trajectory}
from the given stochastic recursive inclusion;
the construction is explained in the next paragraph.
A sequence of `rescaled'
trajectories of length $T$ is constructed as follows: At the beginning of each $T$-length 
interval we observe the trajectory to see if it is outside the unit ball, if so we scale it
back to the boundary of the unit ball. This scaling factor is then used to scale the `rest of
the $T$-length trajectory'.
\\ \indent
To show that the iterates are bounded almost surely we need to show
that the linearly interpolated trajectory does not `run off' to infinity.
To do so we assume that this is not true and show that
there exists a subsequence of the rescaled $T$-length trajectories that has a solution to
$\dot{x}(t) \in h_{\infty}(x(t))$ as a limit point in $C([0,T], \mathbb{R}^d)$.
We choose and fix $T$ such that any solution to $\dot{x}(t) \in h_\infty(x(t))$ with an
initial value inside the unit ball is close to the origin at the end of time $T$. In this paper 
we choose $T = T(\delta_2 - \delta_1) + 1$. We then argue that the linearly interpolated 
trajectory is forced to make arbitrarily large `jumps' within time $T$. The
\textit{Gronwall inequality} is then used to show that this is not possible.
\\ \indent
Once we prove stability of the recursion we invoke
\textit{Theorem 3.6 \& Lemma 3.8} from \textbf{Bena\"{i}m, Hofbauer and Sorin} \cite{Benaim05}
to conclude that the limit set is a closed, connected, internally chain transitive and
invariant set associated with $\dot{x}(t) \in h_\infty(x(t))$.
\paragraph{}
 We construct
the linearly interpolated trajectory $\overline{x}(t)$, for $t \in [0, \infty)$,
from the sequence $\{ x_{n} \}$ as follows: 
Define $t(0) := 0$, $t(n) \ := \ \sum_{i=0}^{n-1} a(i)$. Let
$\overline{x}(t(n)) \ := \ x_{n}$ and for 
$t \ \in \ (t(n), t(n+1))$, let 
\begin{equation}\nonumber
\overline{x}(t) \ := \ 
\left( \frac{t(n+1) - t}{t(n+1) - t(n)} \right)\ \overline{x}(t(n))\ +\ 
\left( \frac{t - t(n)}{t(n+1) - t(n)} \right) \  
\overline{x}(t(n+1)). 
\end{equation}
We define a piecewise constant trajectory using the sequence $\{y_{n}\}_{n \ge 0}$
as follows: $\overline{y(t)} := y_{n}$ for $t \in [t(n), t(n+1))$, $n \ge 0$.
\paragraph{}
We know that the $DI$ given by $\dot{x}(t) \in h_{\infty}(x(t))$ has an
attractor set $\mathcal{A}$ such that 
$\delta_1 := \underset{x \in \mathcal{A}}{\sup} \lVert x \rVert < 1$. Let us fix $T := T(\delta_2 - \delta_1) + 1$, where $T(\delta_2 - \delta_1)$ is as defined in section~\ref{definitions}.
Then, $\lVert x(t) \rVert < \delta_2$, for all $t \ge T(\delta_2 - \delta_1)$, where
$\{x(t): t \in [0, \infty)\}$ is a solution to 
$ \dot{x}(t) \in h_{\infty}(x(t))$ with an initial value inside the unit ball around the origin.
\paragraph{}
Define $T_{0} \ := \ 0$ and $T_{n} \ := \ min\{ t(m) \ : \ t(m) \ge T_{n-1}+T\}$, $n \ge 1$.
Observe that there exists a subsequence $\{m(n)\}_{n \ge 0}$ of $\{n\}$
such that $T_{n} = t(m(n))$ $\forall \ n \ge 0$.
We construct the rescaled trajectory, $\hat{x}(t)$, $t \ge 0$, as follows: Let
$t \in [T_{n}, T_{n+1})$ for some $n \ge 0$,
then
$\hat{x}(t) := \frac{\overline{x}(t)}{r(n)}$, where 
$r(n) = \lVert \overline{x}(T_{n}) \rVert \vee 1$.
Also, let $\hat{x}(T_{n+1}^{-})\ := $  
$\underset{t \uparrow T_{n+1}}{\lim} \hat{x}(t)$, $t \in \left[ T_{n}, T_{n+1} \right)$.
The corresponding `rescaled $y$ iterates' are given by $\hat{y}(t) := 
\frac{\overline{y(t)}}{r(n)}$ and 
the rescaled martingale noise terms by
$\hat{M}_{k+1} := \frac{M_{k+1}}{r(n)}$, $t(k) \in [T_{n}, T_{n+1})$, $n \ge 0$.
\paragraph{}
Consider the recursion at hand, \textit{i.e.,}
\begin{equation}\nonumber
 \overline{x}(t(k+1))\ = \ \overline{x}(t(k)) \ + \ a(k) \left( \overline{y}(t(k))
\ + \ M_{k+1} \right),
\end{equation}
such that $t(k),\ t(k+1)\ 
\in \ \left[ T_{n}, T_{n+1}  \right)$.
Multiplying both sides by $1/r(n)$ we get the rescaled recursion:
\begin{equation}\nonumber
 \hat{x}(t(k+1))\ = \ \hat{x}(t(k)) \ + \ a(k) \left( \hat{y}(t(k))
\ + \ \hat{M}_{k+1} \right).
\end{equation}
Since $\overline{y}(t(k)) \in h 
\left( \overline{x}(t(k)) \right)$, it follows that $\hat{y}(t(k)) \in h_{r(n)}\left( 
\hat{x}(t(k))  \right)$. It is worth noting that
$E \left[ \lVert \hat{M}_{k+1} \rVert ^{2} | \mathcal{F}_{k} \right]$
$\le \ K \left( 1 + \lVert \hat{x}(t(k)) \rVert ^{2} \right)$.
\newpage
\subsection{Characterizing limits, in $C([0,T], \mathbb{R}^d)$,  of the rescaled trajectories}\label{lemmas}
\paragraph{}
We do not provide proofs for the first three lemmas since
they can be found in \textbf{Borkar} \cite{BorkarBook} or \textbf{Bena\"{i}m, Hofbauer and Sorin }
\cite{Benaim05}.
The first two lemmas essentially state that the rescaled martingale noise converges almost surely.
\begin{lemma}
\label{noiseprelim}
 $\underset{t \in [0,T]}{sup}$ $E \lVert \hat{x}(t) \rVert ^{2}$ $< \infty$.
\end{lemma}
\begin{lemma}
\label{noiseconvergence}
 The rescaled sequence $\{ \hat{\zeta}_n\}_{n \ge 1}$, 
 where $\hat{\zeta}_{n}$ $= \sum_{k=0}^{n-1} a(k) \hat{M}_{k+1}$,
 is convergent almost surely.
\end{lemma}
The rest of the lemmas are needed to prove the stability theorem, Theorem~\ref{stability}. 
We begin by showing that the rescaled trajectories are bounded almost surely.
\begin{lemma}
 \label{noescape}
 $\underset{t \in [0, \infty)}{sup} \lVert \hat{x}(t) \rVert < \infty$
 $a.s.$
\end{lemma}
As stated earlier we omit the proof of the above stated lemma and establish a 
 couple of notations used later.
 Let $A \ =\ \{ \omega \ | \ \{\hat{\zeta}_{n}(\omega)\}_{n \ge 1} \ converges\}$.
Since $\hat{\zeta}_{n}$, $n \ge 1$, converges on $A$,
there exists $M_{\omega} < \infty$, possibly sample path dependent, such that 
 $\lVert \sum_{l=0}^{k-1} a(m(n)+l) \hat{M}_{m(n)+l+1} \rVert \le M_{w}$, where $M_{\omega}$
 is independent of $n$ and $k$.
 Also, let
 $\underset{t \ge 0}{\sup}\lVert \hat{x}(t) \rVert$ $\le \ K_\omega$, where $K_\omega := \left( 1 + M_{\omega} + (T+1)K \right)
 e^{K(T+1)}$ is also a constant that is sample path dependent. 
\paragraph{}
Let $x^n (t)$, $t \in [0,T]$ be
the solution (upto time $T$) to
$\dot{x}^{n}(t) = \hat{y}(T_{n} + t)$, with the initial condition 
$x^{n}(0) = \hat{x}(T_{n})$. Clearly, we have
\begin{equation}
 x^{n}(t)\ = \ \hat{x}(T_{n}) + \int_{0}^{t} \hat{y}(T_{n}+z) \,dz .
\end{equation}

The following two lemmas
are inspired by ideas from \textbf{Bena\"{i}m, Hofbauer and Sorin} \cite{Benaim05}
as well as \textbf{Borkar} \cite{BorkarBook}. In the lemma that follows we show
that the limit sets of $\{x^n(\cdotp)\ |\ n \ge 0\}$ and
$\{\hat{x}(T_n + \cdotp)\ |\ n \ge 0 \}$ coincide. We seek limits in
$C([0,T], \mathbb{R}^d)$.
\begin{lemma}
\label{difftozero}
 $\underset{n \to \infty}{\lim}$ 
 $\underset{t \in [T_{n}, T_{n}+T]}{sup} \lVert
 x^n (t) - \hat{x}(t) \rVert = 0$ $a.s.$
\end{lemma}

\begin{proof}
 Let $t \in \left[ t(m(n) + k), t(m(n) + k+1) \right)$ and $t(m(n)+k+1) \le T_{n+1}$.
 We first assume that \\ $t(m(n)+k+1) < T_{n+1}$. We have the following:
 \begin{multline}\nonumber
  \hat{x}(t) = \left( \frac{t(m(n)+k+1) - t}{a(m(n)+k)} \right) \hat{x}(t(m(n)+k))
  +
  \left( \frac{t - t(m(n)+k)}{a(m(n)+k)} \right) \hat{x}(t(m(n)+k+1)).
 \end{multline}
Substituting for $\hat{x}(t(m(n)+k+1))$ in the above equation we get:
\begin{multline}\nonumber
\hat{x}(t) = \left( \frac{t(m(n)+k+1) - t}{a(m(n)+k)} \right) \hat{x}(t(m(n)+k))
  +
  \left( \frac{t - t(m(n)+k)}{a(m(n)+k)} \right) 
  \\
  \left( \hat{x}(t(m(n)+k)) + a(m(n)+k) \left(
  \hat{y}(t(m(n)+k)) + \hat{M}_{m(n)+k+1} \right) \right),
\end{multline}
hence,
 \[
\hat{x}(t) =  \hat{x}(t(m(n)+k)) + 
  \left( t - t(m(n)+k) \right) \left( \hat{y}(t(m(n)+k)) + \hat{M}_{m(n)+k+1} \right).
 \]
Unfolding $\hat{x}(t(m(n)+k))$ over $k$
we get,
 \begin{multline}\label{eq:lemma41}
  \hat{x}(t) = \hat{x}(T_{n})  +  \sum_{l=0}^{k-1} a(m(n)+l)
 \left(  \hat{y}(t(m(n) + l)) + \hat{M}_{m(n)+l+1}\right) + \\
\left(t-t(m(n)+k) \right)
 \left( \hat{y}(t(m(n)+k)) + \hat{M}_{m(n)+k+1}  \right). 
 \end{multline}
 Now, we consider $x^n (t)$, \textit{i.e.,}
 \begin{equation}\nonumber
 x^{n}(t) = \hat{x}(T_{n}) + \int_0^t \hat{y}(T_{n}+z) \ \,dz .
 \end{equation}
 Splitting the above integral, we get
 \begin{equation}\nonumber
  x^n (t) =  \ \hat{x}(T_{n}) \ +\ \sum_{l=0}^{k-1}
 \int_{t(m(n)+l)}^{t(m(n)+l+1)} \hat{y}(z) \,dz
 + \int_{t(m(n)+k)}^{t} \hat{y}(z) \,dz .
 \end{equation}
 Thus,
\begin{multline}\label{eq:lemma42}
 x^n(t) = \hat{x}(T_{n})  +  \sum_{l=0}^{k-1} a(m(n)+l)
 \hat{y}(t(m(n) + l)) + 
\left(t-t(m(n)+k) \right) \hat{y}(t(m(n)+k)). 
\end{multline}
From (\ref{eq:lemma41}) and (\ref{eq:lemma42}), 
 it follows that
 \[
 \lVert x^{n}(t) - \hat{x}(t) \rVert \le 
  \left\lVert \sum_{l=0}^{k-1} a(m(n)+l) \hat{M}_{m(n)+l+1} \right\rVert
 + \left\lVert \left(t-t(m(n)+k) \right) \hat{M}_{m(n)+k+1} \ \right\rVert ,
 \]
 and hence,
 \begin{equation}\nonumber
\lVert x^n (t) - \hat{x}(t) \rVert \le \lVert 
\hat{\zeta}_{m(n)+k} - \hat{\zeta}_{m(n)} \rVert  + 
\lVert \hat{\zeta}_{m(n)+k+1} - \hat{\zeta}_{m(n)+k} \rVert .
 \end{equation}
If $t(m(n)+k+1) = T_{n+1}$ then in the proof we may replace $\hat{x}(t(m(n)+k+1))$ with
$\hat{x}(T^-_{n+1})$. The arguments remain the same.
Since $\hat{\zeta}_{n }$, $n \ge 1$, converges almost surely, the desired result follows.
 \end{proof}
  \paragraph{}
The sets $\{x^n (t), t \in [0,T] \ |\ n\ge 0 \}$ and $\{\hat{x}(T_{n} + t), t \in [0,T]\ |\ n \ge 0\}$
can be viewed as subsets of $C([0,T], \mathbb{R}^{d})$. It can be shown that 
$\{x^n (t), t\in [0,T] \ |\  n \ge 0\}$ is equi-continuous and point-wise bounded. Thus from
the \textit{Arzela-Ascoli} theorem, $\{x^n (t), t \in [0,T] \ |\ n\ge 0 \}$ 
is relatively compact.
It follows from Lemma~\ref{difftozero} that the set $\{\hat{x}(T_{n} + t), t \in [0,T] \ |\ n \ge 0\}$
is also relatively compact in $C([0,T], \mathbb{R}^{d})$.
\begin{lemma}
\label{closertoode}
 Let $r(n) \uparrow \infty$, then any limit point of $\left\{ \hat{x}(T_{n}+t), t \in [0,T] 
 : n \ge 0 \right\}$ is of the form $x(t) = x(0) +  \int_0^t y(s) \ ds$, where 
 $y: [0,T] \rightarrow \mathbb{R}^{d}$ is a measurable function and
 $y(t) \in h_{\infty}(x(t))$, $t \in [0,T]$.
\end{lemma}
\begin{proof}
We define the notation $[t] := max\{t(k) \ |\ t(k) \le t\}$, $t \ge 0$.
Let $t \in [T_{n}, T_{n+1})$, then $\hat{y}(t) \in h_{r(n)}(\hat{x}([t]))$
and $\lVert \hat{y}(t) \rVert$ $\le $ $K \left( 1 +  \lVert \hat{x}([t]) \rVert\right)$
since $h_{r(n)}$ is a Marchaud map ($K$ is the constant associated with
the point-wise boundedness property). It follows from Lemma~\ref{noescape}
that $\underset{t \in [0, \infty)}{sup}$ $\lVert \hat{y}(t) \rVert < \infty$ $a.s.$
 Using observations made earlier, we can deduce that there exists a sub-sequence of 
 $\mathbb{N}$, say
 $\{ l\} \subseteq \{n\}$, such that $\hat{x}(T_{l}+t) \to x(t)$ in $C \left([0,T], 
 \mathbb{R}^{d} \right)$ and $\hat{y}(m(l)+ \cdotp) \to y(\cdotp)$ weakly in $L
 _{2} \left( [0,T], \mathbb{R}^{d} \right)$. From Lemma~\ref{difftozero}
 it follows that
 $x^{l}(\cdotp) \to x(\cdotp)$ in $C \left( [0,T], \mathbb{R}^{d} \right)$. Letting $r(l) 
 \uparrow \infty$ in
 \begin{equation}\nonumber
 x^{l}(t) = x^{l}(0) + \int_{0}^{t} \hat{y}(t(m(l) + z)) \,dz , \ t \ \in \ [0, T],
 \end{equation}
 we get $x(t) = x(0) + \int_{0}^{t} y(z) dz$ for $t \in [0,T]$. 
Since $\lVert \hat{x}(T_{n})
 \rVert \ \le \ 1$ we have $\lVert x(0) \rVert \ \le \ 1$.
\paragraph{}
 Since $\hat{y}(T_{l}+ \ \cdotp) \to y(\cdotp)$ weakly in 
 $L_{2} \left( [0,T], \mathbb{R}^{d} \right)$,
 there exists $\{l(k)\} \subseteq \{l\}$
 such that \[\frac{1}{N} \sum_{k=1}^{N} \hat{y}(T_{l(k)}+ \ \cdotp) \to y(\cdotp)
 \text{ strongly in }L_{2} \left( [0,T], \mathbb{R}^{d} \right).\]
  Further, there exists
 $\{N(m)\} \subseteq \{N\}$ such that 
 \[\frac{1}{N(m)} \sum_{k=1}^{N(m)} \hat{y}(T_{l(k)}+ \ \cdotp) \to y(\cdotp) \text{ \textit{a.e.} on } 
 [0,T].\]
 \paragraph{}
 Let us fix $t_0 \in$ $\{ t \ |\ \frac{1}{N(m)} \sum_{k=1}^{N(m)} \hat{y}(T_{l(k)}+ \ t) \to y(t),\ t \in [0,T] \}$,
 then
 \begin{equation}\nonumber
  \lim_{N(m) \to \infty} \frac{1}{N(m)} \sum_{k=1}^{N(m)} \hat{y}(T_{l(k)}+ \ t_0) = y(t_0).
 \end{equation}
Since $h_{\infty}(x(t_0))$ is convex and compact (Proposition~\ref{propo}), 
to show that $y(t_0) \in h_{\infty}(x(t_0))$ it is enough to prove that
$
\underset{l(k) \to \infty}{\lim} d\left(\hat{y}(T_{l(k)} +  t_0), h_{\infty}(x(t_0))\right) = 0.
$
If not, $\exists$ $\epsilon > 0$ and $\{n(k)\} \subseteq \{l(k)\}$ such that
$d\left(\hat{y}(T_{n(k)} +  t_0), h_{\infty}(x(t_0))\right) > \epsilon$. Since
$\{\hat{y}(T_{n(k)} +  t_0)\}_{k \ge 1}$ is norm bounded, it follows that there is 
a convergent sub-sequence. For the sake of convenience we assume that
$\underset{k \to \infty}{\lim}$ $ \hat{y}(T_{n(k)} +  t_0)  = y$, for some
$y \in \mathbb{R}^{d}$. Since $\hat{y}(T_{n(k)} +  t_0) \in h_{r(n(k))}(\hat{x}([T_{n(k)} +  t_0]))$
and $\underset{k \to \infty}{\lim}$ $ \hat{x}([T_{n(k)} +  t_0])  = x(t_0)$, it
follows from assumption $(A5)$ that $y \in h_{\infty}(x(t_0))$. 
This leads to a contradiction.
 \end{proof}
 \paragraph{}
 Note that in the statement of Lemma~\ref{closertoode} 
 we can replace `$r(n) \uparrow \infty$' by `$r(l) \uparrow \infty$',
 where $\{ r(l)) \}$ is a subsequence of $\{ r(n) \}$.
 Specifically we can conclude that
 any limit point of $\left\{ \hat{x}(T_{k}+t), t \in [0,T] \right\} _{
  \{k\} \subseteq \{ n \} }$ in $C([0,T], \mathbb{R}^d)$, conditioned on $r(k) \uparrow \infty$, 
  is of the form $x(t) = x(0) +  \int_0^t y(z) \,dz$, where $y(t) \in
 h_{\infty}(x(t))$ for $t \in [0,T]$. It should be noted that $y(\cdotp)$
 may be sample path dependent. The following is an immediate consequence
 of Lemma~\ref{closertoode}.
 \begin{corollary} \label{r_0}
  $\exists \ 1 < R_0 < \infty$ 
such that $\forall \ r(l) > R_0$
$\lVert \hat{x}(T_l + \cdotp) - x(\cdotp) \rVert < \delta_3 - \delta_2$,
where $\{ l \} \subseteq \mathbb{N}$ and $x(\cdotp)$ is a solution (up to time $T$) of $\dot{x}(t) \in h_\infty(x(t))$
such that $\lVert x(0) \rVert \le 1$. The form of $x(\cdotp)$ is as given by
Lemma~\ref{closertoode}.
 \end{corollary}
\begin{proof}
Assume to the contrary that $\exists \ r(l) \uparrow \infty$ such that
$\hat{x}(T_l + \cdotp)$ is at least $\delta_3 - \delta_2$ away from any solution
to the $DI$. It follows from Lemma~\ref{closertoode} that
there exists a subsequence of 
$\{ \hat{x}(T_l + t), 0 \le t \le T \ :\ l \subseteq \mathbb{N} \}$
guaranteed to converge, in $C([0,T], \mathbb{R}^d)$,
to a solution of $\dot{x}(t) \in h_\infty (x(t))$ such that $\lVert x(0) \rVert \le 1$. 
This is a contradiction.
 \end{proof}
 \paragraph{}
 \textit{It is worth noting that $R_0$ may be sample path dependent.
Since $T = T(\delta_2 - \delta_1) +1$ we get $\lVert \hat{x}([T_l + T]) \rVert < \delta_3$
for all $T_l$ such that $\lVert \overline{x}(T_l)\rVert (=r(l)) > R_0$.}
\subsection{Stability theorem}\label{mainresult}
\paragraph{}
We are now ready to prove the stability
of a \textit{SRI} given by (\ref{eq:recursiveinclusion}) under the assumptions
$(A1) - (A5)$.
If $\underset{n}{sup}$ $r(n) < \infty$, then the iterates are stable and there is nothing
to prove. If on the other hand $\underset{n}{sup}$ $r(n) = \infty$, there
exists $\{l\} \subseteq \{n\}$ such that $r(l) \uparrow \infty$.
It follows from Lemma~\ref{closertoode}  that any limit point of
$\left\{ \hat{x}(T_{l}+t), t \in [0,T] 
 : \{l\} \subseteq \{n\} \right\}$ is of the form $x(t) = x(0) +  \int_0^t y(s) \ ds$, 
 where $y(t) \in h_{\infty}(x(t))$ for $t \in [0,T]$. 
 From assumption $(A4)$, we have that $\lVert x(T) \rVert < \delta_2$.
 Since the time intervals are roughly $T$ apart, for large values of $r(n)$ we can conclude
 that $\lVert \hat{x} \left(T_{n+1}^{-} \right) \rVert < \delta_3$, where
 $\hat{x}(T_{n+1}^{-}) \ =$ $\lim_{t \uparrow t(m(n+1))} \hat{x}(t)$, $t \in 
\left[T_{n}, T_{n+1} \right)$.
 \begin{theorem}[Stability Theorem for DI]
\label{stability}
  Under assumptions $(A1) - (A5)$, $\underset{n}{sup} \lVert x_{n} \rVert < \infty$ \textit{a.s.}
\end{theorem} 

\begin{proof}
 As explained earlier it is sufficient to consider the case when 
 $ \underset{n}{sup} \ r(n) = \infty$. 
 Let $\{l\} \subseteq \{n\}$ such that $r(l) \uparrow \infty$.
 Recall that $T_l = t(m(l))$ and that $[T_{l}+T] = max\{t(k) \ |\ t(k) \le T_{l}+T \}$.
 \paragraph{}
We have $\lVert x(T) \rVert < \delta_2$ since $x(t)$ is a solution, up to time $T$, 
to the $DI$ given by $\dot{x}(t) \in h_{\infty}(x(t))$ and we have fixed $T = T(\delta_2 - \delta_1) + 1$.
From Lemma~\ref{closertoode} we conclude that 
there exists $N$ such that all of the following happen: 
\begin{itemize}
 \item[(i)]
$m(l) \ge N$ $\implies$ $\lVert \hat{x}([T_{l}+T]) \rVert < \delta_3$.
\item[(ii)]  $n \ge N$ $\implies$
$a(n) < \frac{\delta_4 - \delta_3}{\left[K(1+K_{\omega})+M_{\omega}\right]}$.
\item[(iii)]  $n \ > \ m \ \ge \ N$ $\implies$
$\lVert \hat{\zeta}_{n} \ - \ \hat{\zeta}_{m} \rVert < M_{\omega}$.
\item[(iv)] $m(l) \ge N$ $\implies$  $r(l) \ > \ R_0$.
\end{itemize}
In the above, $R_0$ is defined in the statement of Corollary~\ref{r_0} and
$K_{\omega}$, $M_{\omega}$ are explained in Lemma~\ref{noescape}.
\paragraph{}
Recall that we chose 
$\underset{x \in \mathcal{A}}{\sup} \lVert x \rVert = \delta_1 < \delta_2 < \delta_3 < \delta_4 <1$ in Section~\ref{assumptions}.
Let $m(l) \ge N$ and $t(m(l+1)) \ = \ t(m(l) + k + 1)$ for some $k \ge 0$.
Clearly from the manner in which the $T_{n}$ sequence is defined, we have
$t(m(l)+k) \ = \ \left[ T_{l}+T \right]$. 
As defined earlier $\hat{x}(T_{n+1}^{-}) \ =$ $\lim_{t \uparrow t(m(n+1))} \hat{x}(t)$, $t \in 
\left[T_{n}, T_{n+1} \right)$ and $n \ge 0$.
Consider the equation
\begin{equation}\nonumber
 \hat{x}(T_{l+1}^{-})\ = \ \hat{x}(t(m(l)+k)) \ +\ a(m(l)+k) \left( 
\hat{y}(t(m(l)+k))+ \hat{M}_{m(l)+k+1} \right).
\end{equation}
Taking norms on both sides we get,
\[
\lVert \hat{x}(T_{l+1}^{-}) \rVert\ \le \  
\lVert \hat{x}(t(m(l)+k)) \rVert \ +\ 
a(m(l)+k) \lVert 
\hat{y}(t(m(l)+k)) \rVert \ + \ a(m(l)+k)\lVert \hat{M}_{m(l)+k+1} \rVert. 
\]
From the way we have chosen $N$ we conclude that:
\[
 \lVert \hat{y}(t(m(l)+k)) \rVert \ \le \ K \left( 1 + \lVert \hat{x}(t(m(l)+k) \rVert \right)
\ \le \ K \left( 1 + K_{\omega} \right) \ and\ that
\]
\begin{equation}\nonumber
\lVert \hat{M}_{m(l)+k+1} \rVert = \ \lVert \hat{\zeta}_{m(l)+k+1} 
- \hat{\zeta}_{m(l)+k}\rVert \ \le M_{\omega}. 
\end{equation}
Thus we have that,
\begin{equation}\nonumber
\lVert \hat{x}(T_{l+1}^{-}) \rVert \ \le \ 
\lVert \hat{x}(t(m(l)+k))\rVert \ + \ a(m(l)+k) \left( K(1+K_{\omega}) + M_{\omega} \right).
\end{equation}
Finally we have that $\lVert \hat{x}(T_{l+1}^{-}) \rVert \ < \ \delta_4$ and
\begin{equation} \label{eq:mainthm}
 \frac{\lVert \overline{x}(T_{l+1}) \rVert}{\lVert \overline{x}(T_{l}) 
 \rVert} \ = \ \frac{\lVert \hat{x}(T_{l+1}^{-}) \rVert}{\lVert \hat{x}(T_{l}) 
 \rVert} \ < \delta_4 < 1. 
\end{equation}

It follows from (\ref{eq:mainthm}) that $\lVert \overline{x}(T_{n+1}) \rVert < 
\delta_4 \lVert \overline{x}(T_n) \rVert $ if $\lVert \overline{x}(T_n) \rVert > R_0$.
From Corollary~\ref{r_0} and the aforementioned we get that
the trajectory falls at an exponential rate till it enters $\overline{B}_{R_{0}}(0)$. 
Let $t \le T_{l}$, $t \in \left[T_{n}, T_{n+1}\right)$ and $n+1 \le l$, 
be the last time that $\overline{x}(t)$ jumps from $\overline{B}_{R_{0}}(0)$ to
the outside of the ball. It follows that 
$\lVert \overline{x}(T_{n+1}) \rVert \ge \lVert \overline{x}(T_l) \rVert$.
Since $r(l) \uparrow \infty$, $\overline{x}(t)$ would be forced to make 
larger and larger jumps within an interval of $T+1$.
This leads to a contradiction since the maximum jump within any fixed time interval
can be bounded using the \textit{Gronwall} inequality.  
 \end{proof}
\paragraph{}
We now state one of the main theorems of this paper. 

\begin{theorem}
\label{main}
Under assumptions $(A1)-(A5)$,
almost surely, the sequence $\{x_{n}\}_{n \ge 0}$ generated by the
stochastic recursive inclusion, given by (\ref{eq:recursiveinclusion}), is bounded
and converges to a closed, connected, internally chain transitive and invariant
set of $\dot{x}(t) \in h(x(t))$.
\end{theorem}
\begin{proof}
 The stability of the iterates is shown in Theorem~\ref{stability}.
 The convergence can be proved under assumptions $(A1)-(A3)$ 
 and the stability of the iterates
in exactly the same manner as in \textit{Theorem 3.6 \& Lemma 3.8} 
of \textbf{Bena\"{i}m, Hofbauer and Sorin} \cite{Benaim05}.
 \end{proof}
  \paragraph{}
We have thus far shown that under assumptions $(A1)-(A5)$ the $SRI$ given by ($\ref{eq:recursiveinclusion}$)
is stable and converges to a closed, connected, internally chain transitive and invariant set.
\newpage
\subsection{Stability theorem under modified assumptions}\label{AccuMainSec}
\paragraph{}
In $(A4)$ we assumed that $Liminf_{c \to \infty} h_{c}(x)$ is nonempty 
for all $x \in \mathbb{R}^d$.
In this section we shall develop a stability criterion for the case when we no longer make
such an assumption. In other words, we work with a modified version of assumption $(A4)$
that we call $(A6)$.
\subsection*{Modification of Assumption (A4)}
\paragraph{}
Recall the following \textit{SRI}:
\begin{equation}\label{eq:recursiveinclusion1}
  x_{n+1}\ = x_{n}\ +\ a(n) \left[ y_{n}\ +\ M_{n+1} \right],\ \text{for }n \ge 0.
\end{equation}
Since $h_c$ is point-wise bounded for each $c \ge 1$, we have
$\underset{y \in h_c (x)}{sup}$ $\lVert y \rVert$ $\le K(1+ \lVert x \rVert)$, where $x \in \mathbb{R}^d$
(see Proposition~\ref{propo}).
This implies that $\{y_c\}_{c \ge 1}$, where $y_c \in h_c (x)$,
has at least one convergent subsequence. It follows from the definition of
\textit{upper-limit} of a sequence of sets (see Section~\ref{definitions})
that $Limsup_{c \to \infty} h_c (x)$ is non-empty for every $x \in \mathbb{R}^d$. 
It is worth noting that $Liminf_{c \to \infty} h_c(x)$ $\subseteq \ 
Limsup_{c \to \infty} h_c(x)$ for every $x \in \mathbb{R}^d$. 
Another important point to consider is that the lower-limits
of sequences of sets are harder to compute than their upper-limits,
see \textbf{Aubin} \cite{AubinSet} for more details.
\paragraph{}
Recall that 
 $h_{c}(x)$ $= \ \{ y \ | \ cy \in h(cx) \}$, where $x \in \mathbb{R}^d$
 and $c \ge 1$. Clearly the upper-limit,
$\ Limsup_{c \to \infty } \ h_{c}(x)
= \{y \ |\ \underset{c \to \infty}{\underline{lim}}d(y, h_c (x))= 0 \}$ is nonempty
for every $x \in \mathbb{R}^{d}$. 
For $A \subseteq \mathbb{R}^d$,  $\overline{co}(A)$ 
denotes the closure of the convex hull of $A$, $i.e.,$ the 
closure of the smallest convex set containing $A$. 
\[ \text{Define }
h_{\infty}(x) := \overline{co} \left( \ Limsup_{c \rightarrow \infty} \ h_{c}(x) \right).\]
Below we state the modification of assumption $(A4)$ that we call $(A6)$.
\\
\begin{itemize}
\item[\textbf{(A6)}] \textit{The differential inclusion $\dot{x}(t) \ \in \ h_{\infty}(x(t))$ has an attracting set
 $\mathcal{A} \subset B_{1}(0)$ and $\overline{B}_{1}(0)$ is a subset of some fundamental neighborhood
 of $\mathcal{A}$.}
\end{itemize}
$\\$
Note that in 
$(A4)$, $h_{\infty}(x) := \overline{\ Liminf_{c \rightarrow \infty } \ h_{c}(x)}$
while in
$(A6)$, $h_{\infty}(x) := \overline{co} \left( \ Limsup_{c \rightarrow \infty} \ h_{c}(x) \right)$.
In this section we shall work with this new definition of $h_\infty$.
\begin{proposition}\label{propo1}
 $h_{\infty}$ is a Marchaud map.
\end{proposition}
\begin{proof}
 From the definition of $h_{\infty}$ 
 it follows that $h_\infty(x)$ is convex, compact for all $x \in \mathbb{R}^d$ and
 $h_\infty$ is point-wise bounded.
It is left to prove that $h_{\infty}$ is an upper-semicontinuous map.
 \paragraph{}
 Let $x_{n} \to x$, $y_{n} \to y$ and $y_{n} \in h_{\infty}(x _n)$, 
 for all $n \ge 1$.
 We need to show that $y \in h_{\infty}(x)$. We present a proof by contradiction.
 Since $h_{\infty}(x)$ is convex and compact, $y \notin h_{\infty}(x)$
 implies that there exists a linear functional on $\mathbb{R}^{d}$, say $f$, such that
 $\underset{z \in h_{\infty}(x)}{sup}$ $f(z) \le \alpha - \epsilon$
 and $f(y) \ge \alpha + \epsilon$, for some
 $\alpha \in \mathbb{R}$ and $\epsilon > 0$. Since $y_{n} \to y$, there exists 
 $N > 0$ such that for all $n \ge N$, $f(y_{n}) \ge \alpha + \frac{\epsilon}{2}$. In other
 words, $h_{\infty}(x) \cap  [f \ge \alpha + \frac{\epsilon}{2}] \neq \phi$ for
 all $n \ge N$. We use the notation $[f \ge a]$ to denote the set
 $\left\{ x \ |\ f(x) \ge a \right\}$. For the sake of convenience let us denote the set
 $Limsup_{c \to \infty}h_{c}(x)$ by $A(x)$, where $x \in \mathbb{R}^{d}$.
 We claim that $A(x_{n}) \cap [f \ge \alpha + \frac{\epsilon}{2}] \neq \phi$
 for all $n \ge N$. We prove this claim later,
 for now we assume that the claim
 is true and proceed. Pick $z_{n} \in A(x_{n}) \cap [f \ge \alpha + \frac{\epsilon}{2}]$
 for each $n \ge N$. It can be shown that $\{z_{n}\}_{n \ge N}$ is norm bounded
 and hence contains a convergent subsequence, 
 $\{z_{n(k)}\}_{k \ge 1} \subseteq \{z_{n}\}_{n \ge N}$. 
 Let $\underset{k \to \infty}{\lim} z_{n(k)} = z$.
Since $z_{n(k)} \in Limsup_{c \to \infty}(h_{c}(x_{n(k)}))$, 
 $\exists$ $c_{n(k)} \in \mathbb{N}$ such that $\lVert w_{n(k)} - z_{n(k)} \rVert
 < \frac{1}{n(k)}$, where $w_{n(k)} \in h_{c_{n(k)}}(x_{n(k)})$. 
 We choose the sequence
 $\{c_{n(k)}\}_{k \ge 1}$ such that $c_{n(k+1)} > c_{n(k)}$ for each $k \ge 1$.
 \\ \indent
 We have the following: $c_{n(k)} \uparrow \infty$, $x_{n(k)} \to x$, 
 $w_{n(k)} \to z$ and $w_{n(k)} \in h_{c_{n(k)}}(x_{n(k)})$, for all $ k \ge 1$. 
 It follows
 from assumption $(A5)$ that $z \in h_{\infty}(x)$. Since $z_{n(k)} \to z$
 and $f(z_{n(k)}) \ge \alpha + \frac{\epsilon}{2}$ for each $k \ge 1$, we have that
 $f(z) \ge \alpha + \frac{\epsilon}{2}$. This contradicts the earlier conclusion that
 $\underset{z \in h_{\infty}(x)}{sup}$ $f(z) \le \alpha - \epsilon$.
 \paragraph{}
 It remains to prove that  $A(x_{n}) \cap [f \ge \alpha + \frac{\epsilon}{2}] \neq \phi$
 for all $n \ge N$. If this were not true, then
 $\exists \{m(k)\}_{k \ge 1} \subseteq \{n \ge N\}$ 
 such that $A(x_{m(k)}) \subseteq [f < \alpha + \frac{\epsilon}{2}]$
 for all $k$. It follows that \\
$h_\infty(x_{m(k)}) = \overline{co}(A(x_{m(k)})) \subseteq 
 [f \le \alpha + \frac{\epsilon}{2}]$ for each $k \ge 1$. 
 Since $y_{n(k)} \to y$, $\exists N_{1}$ such that for all $n(k) \ge N_1$, 
 $f(y_{n(k)}) \ge \alpha + \frac{3 \epsilon}{4}$. This is a contradiction.
\end{proof}
\paragraph{}
We are now ready to state the second stability theorem for an \textit{SRI} given by
(\ref{eq:recursiveinclusion1}) under a modified set of
assumptions. We retain assumptions $(A1)-(A3)$, replace $(A4)$ by $(A6)$
and finally in $(A5)$ we let $h_{\infty}(x) :=
\overline{co} \left( \ Limsup_{c \rightarrow \infty} \ h_{c}(x) \right)$. We state the
theorem under these updated set of assumptions.
\begin{theorem}[Stability Theorem for DI \#2]\label{AccuMain}
Under assumptions $(A1)-(A3)$, $(A5)$ (with $h_\infty (x) := \overline{co} 
( Limsup_{c \to \infty} h_c (x))$) and $(A6)$,
almost surely the sequence $\{x_{n}\}_{n \ge 0}$ generated by the
stochastic recursive inclusion, given by (\ref{eq:recursiveinclusion1}) is bounded
and converges to a closed, connected internally chain transitive invariant
set of $\dot{x}(t) \in h(x(t))$.
\end{theorem}
\begin{proof}
 The statements of $Lemmas~$ \ref{noiseprelim}$-$\ref{closertoode} hold true even when 
 $h_{\infty} := \overline{co} \left(\ Limsup_{c \rightarrow \infty} \ h_{c}(x) \right)$ 
 and $(A5)$ is interpreted as explained earlier. 
 The stability of the iterates can be proven in an identical manner 
 to the proof of Theorem~\ref{stability}.
Next,
we invoke \textit{Theorem 3.6 \& Lemma 3.8} of
 \textbf{Bena\"{i}m, Hofbauer and Sorin} \cite{Benaim05} to conclude that the iterates converge to a closed,
 connected, internally chain transitive and invariant set of $\dot{x}(t) \in h(x(t))$.
 \end{proof}
  \paragraph{}
\begin{remark} \label{a4'}
 Assumptions $(A4)$ and $(A6)$ required that $\dot{x}(t) \in h_{\infty}(x(t))$ have an attractor
 set inside $B_1 (0)$ (the open unit ball). Further, it required $\overline{B}_1 (0)$ to be
 in its fundamental neighborhood. 
 Note that 
$h_{\infty}(x)$ is defined as $\overline{\ Liminf_{c \rightarrow \infty } \ h_{c}(x)}$
when using $(A4)$ and it is defined as
$\overline{co} \left( \ Limsup_{c \rightarrow \infty} \ h_{c}(x) \right)$
when using $(A6)$.
Consider the following generalization of $(A4)/(A6)$.
 \begin{itemize}
  \item[\textbf{(A4)$'$/(A6)$'$}:] \textit{$\dot{x}(t) \in h_\infty (x(t))$ has an attractor set $\mathcal{A}$ such that
  $\mathcal{A} \subseteq B_a (0)$ and $\overline{B}_a(0)$ is a subset of its fundamental neighborhood,
  where $0 \le a < \infty$.}
 \end{itemize}
Note that $a$ could be greater than $1$, further since $\mathcal{A}$ is compact by definition,
$a$ is finite. A sufficient condition for $(A4)'/(A6)'$ is when $\mathcal{A}$ is a globally attracting,
Lyapunov stable set associated with $\dot{x}(t) \in h_\infty (x(t))$. In this case any compact set
is a fundamental neighborhood of $\mathcal{A}$.
\\ \indent At the beginning of Section~\ref{sec3} we constructed the rescaled trajectory by projecting
onto the unit ball around the origin. In order to use
$(A4)'/(A6)'$ we build the rescaled trajectory by projecting onto $\overline{B}_a(0)$ 
instead.
We can modify the proofs such that the statements of Theorems~\ref{main} and \ref{AccuMain}
remain true under assumptions $(A1)-(A3)$, $(A4)'/(A6)'$ and $(A5)$.
\end{remark}
\paragraph{}
\begin{remark}\label{lyapunova4}
The advantage of using $(A4)'/(A6)'$ is that one can conclude the stability of
the iterates by merely possessing the knowledge that the associated $DI$ of the infinity system
has a global attractor set. Consider the following trivial example
of a stochastic gradient descent algorithm with
linear gradient function of the from $-(Ax +b)$. The corresponding infinity system,
$\dot{x}(t) = -Ax$, is clearly ``related'' to the associated o.d.e. $\dot{x}(t) = -(Ax +b)$. 
Specifically, if there was a unique global minimizer then both the aforementioned o.d.e.'s have a global
attractor which in turn implies the stability of the iterates as discussed before.
This trivial example also illustrates a finer point that $h_\infty$ and $h$ could be related,
hence information about $h$ could help us ascertain if $(A4)'/(A6)'$ is satisfied. 
Whenever possible one could also construct Lyapunov functions to ascertain the same.
While we did not consider Lyapunov-type conditions for stability, it would be interesting
to extend the Lyapunov-type stability conditions developed for $SRE$'s 
by \textbf{Andrieu, Priouret and Moulines} \cite{Andrieu} to include $SRI$'s.
\end{remark}
\newpage
\section{Extensions to the stability theorem of Borkar and Meyn}\label{GenBorkarMeynSec}
\paragraph{}
We begin this section by listing the assumptions (See \textit{Section 2} of \cite{Borkar99}) and 
statement of the \textit{Borkar-Meyn Theorem} (See \textit{Section 2.1} of \cite{Borkar99}).
The notations used are consistent with those of equation~(\ref{eq:basicrecursion}).
\begin{enumerate}
 \item[(BM1)] \label{BM1}
(i) The function $h:\mathbb{R}^d\rightarrow
\mathbb{R}^d$
is Lipschitz continuous, with Lipschitz constant $L$. 
There exists a function $h_\infty:\mathbb{
R}^d \rightarrow \mathbb{R}^d$
such that $\underset{c \to \infty}{\lim} \frac{h(cx)}{c} =
h_\infty(x)$, for each
$x\in \mathbb{R}^d$.\\
(ii) $h_{c} \rightarrow h_{\infty}$ uniformly on compacts, as $c \to \infty$.\\
(iii) The o.d.e. $\dot{x}(t) = h_\infty(x(t))$
has the origin as the unique globally asymptotically stable equilibrium.\\
\item[(BM2)]
$\{ a(n) \}_{n \ge 0}$ is a scalar sequence such that: $a(n) \ge 0$, $\underset{n \ge 0}{\sum} a(n) = \infty$
 and 
 \\
 $\underset{n \ge 0}{\sum} a(n)^{2} < \infty$. Without loss of generality, we also let
 $\underset{n}{sup}\ a(n) \le 1$.
 \item[(BM3)]
 $\{ M_{n}\}_{n \ge 1}$ is a martingale difference sequence with respect to
 the filtration
 $\mathcal{F}_{n}$ $:=$ $ \sigma \left( x_{0}, M_{1}, \ldots, M_{n} \right) $, 
 $n \ge 0$. Thus,
 $E\left[ M_{n+1} | \mathcal{F}_{n} \right] $ $=\ 0 \ a.s.$, $\forall \ n \ge 0$.
$\{ M_{n}\}$
 is also square integrable with
 $E[\lVert M_{n+1} \rVert ^{2} | \mathcal{F}_{n}]$ $\le$ $L \left( 1 + \lVert x_{n} \rVert ^{2} \right)$, for some
 constant $L > 0$. Without loss of generality, assume that the same constant, $L$,
 works for both $(BM1)(i)$ and $(BM3)$.
\end{enumerate}

\begin{theorem}[Borkar-Meyn Theorem]
\label{BorkarMeyn}
Suppose (BM1)-(BM3) hold.
Then
$\underset{n}{sup} \lVert x_n \rVert <\infty$ almost surely.
Further, the sequence $\{x_{n}\}$ converges almost surely to a
(possibly sample path dependent) compact connected internally chain transitive
invariant set of $\dot{x}(t) = h(x(t))$.
\end{theorem}
In what follows we illustrate a weakening of $(BM1)-(BM3)$ stated above using
Theorems~\ref{main} \& \ref{AccuMain}. Note that $(BM2)$ is the standard step-size assumption
while $(BM3)$ is the assumption on the martingale difference noise; we
endeavor to weaken $(BM1)$.
\subsection{Superfluity of (BM1)(ii) as a consequence of Theorem~\ref{main}}\label{BMRelax}
\paragraph{}
In this section we discuss in brief how the \textit{Borkar-Meyn Theorem} (Theorem~\ref{BorkarMeyn})
can be proven under $(BM1)(i),(iii)$, $(BM2)$ and $(BM3)$. In other words, we show that
$(BM1)(ii)$ is superfluous. In this direction we begin by showing the following:
A recursion given by (\ref{eq:basicrecursion}) satisfies $(BM1)(i),(iii)$, $(BM2)$ and $(BM3)$
\begin{Large}$\Rightarrow$ \end{Large} (\ref{eq:basicrecursion}) satisfies $(A1)-(A5)$ of Section~\ref{assumptions}.
The following implications are straightforward: $(BM1)(i),(iii) \Rightarrow (A1) \ \& \ (A4)$; 
$(BM2) \Rightarrow (A2)$; $(BM3) \Rightarrow (A3)$. We now show 
$(BM1)(i),(iii) \Rightarrow (A5)$. Given $x_n \to x$, $c_n \uparrow \infty$
and $h_{c_n}(x_n) \to y$ we need to show $y = h_{\infty}(x)$. We have the following:
\[
\lVert h_{c_n}(x_n) - h_{\infty}(x) \rVert \le \lVert h_{c_n}(x_n) - h_{c_n}(x) \rVert +
\lVert h_{c_n}(x) - h_{\infty}(x) \rVert.
\]
If $h$ is Lipschitz with constant $L$ then it can be shown that
$h_{c}$ ($h_c:x \mapsto \frac{h(cx)}{c}$, $x \in \mathbb{R}^d$) is Lipschitz, for every $c \ge 1$, with the same constant.
Further, $h_{c_n}(x) \to h_{\infty}(x)$ as $c_n \uparrow \infty$. Taking limits ($c_n \uparrow \infty$)
on both sides of the above equation gives $\underset{c_n \uparrow \infty}{\lim} h_{c_n}(x_n) = h_\infty(x)$
as required. Since $(A1)-(A5)$ are satisfied it follows from Theorem~\ref{main} that a $SRE$
satisfying $(BM1)(i), (iii), \ (BM2), \ (BM3)$ is stable and converges to a closed, connected, internally chain transitive and
invariant set of $\dot{x}(t) = h(x(t))$ (Theorem~\ref{BorkarMeyn}).
\paragraph{}
We discuss in brief how we work around using $(BM1)(ii)$ in proving the 
\textit{Borkar-Meyn Theorem}. 
The notations used herein are consistent with
those found in \textit{Chapter 3} of \textbf{Borkar} \cite{BorkarBook}. We list a few below
for easy reference.
\begin{enumerate}
 \item $\phi_n (\cdotp, x)$ denotes the solution to $\dot{x}(t) \in h_{r(n)}(x(t))$
 with initial value $x$.
 \item $\phi_{\infty}(\cdotp, x)$ denotes the solution to $\dot{x}(t) \in h_{\infty}(x(t))$
 with initial value $x$.
 \item $x^n(t)$, $t \in [0,T]$ denotes the solution to $\dot{x}^n(t) = h_{r(n)}(\hat{x}(T_n + t))$
 with initial value $x^n(0) = \hat{x}(T_n)$. Then
 $x^n (t) = \phi_{n}(t, \hat{x}(T_{n}))$, $t \in [0, T]$.
\end{enumerate}
In proving the
\textit{Borkar-Meyn Theorem} 
as outlined in \cite{Borkar99} $(BM1)(ii)$ is used
to show that for large values of $r(n)$, $\phi_n (t, \hat{x}(T_{n}))$ is `close'
to $\phi_{\infty} (t, \hat{x}(T_{n}))$, $t \in [0, T]$.
In this paper we deviate from \cite{Borkar99} in the definition of 
$x^n (t)$, $t \in [0,T]$, here $x^n (\cdotp)$ denotes the solution up to time $T$
of $\dot{x}^n (t) = \hat{y}(T_n + t) = h_{r(n)}(\hat{x}([T_n + t]))$
with $x^n (0) = \hat{x}(T_{n})$,
where $[\cdotp]$ is defined in Lemma~\ref{closertoode}. In other
words, we have the following:
 \begin{equation}\nonumber
  x^n (t) =  \ \hat{x}(T_{n}) \ +\ \sum_{l=0}^{k-1}
 \int_{t(m(n)+l)}^{t(m(n)+l+1)} \hat{y}(z) \,dz
 + \int_{t(m(n)+k)}^{t} \hat{y}(z) \,dz.
 \end{equation}
 For $t \in [t_n, t_{n+1})$, $\hat{y}(t)$ is a constant and equals $\hat{y}(t_n)$. We get the following:
 \begin{multline}\nonumber
 x^n(t) = \hat{x}(T_{n})  +  \sum_{l=0}^{k-1} a(m(n)+l)
 h_{r(n)} \left( \hat{x}([t(m(n) + l)]) \right) + 
\left(t-t(m(n)+k) \right) h_{r(n)} \left( \hat{x}([t(m(n)+k)]) \right). 
\end{multline}
The proof now proceeds along the lines of Section~\ref{mainresult}
\textit{i.e.,}
Lemmas~\ref{noiseprelim} - \ref{closertoode} and Theorem~\ref{stability};
we essentially show the following: If $r(n) \uparrow \infty$ then the $T$-length trajectories given by
$\{x^n (\cdotp)\}_{n \ge 0}$ have $\phi_{\infty}(x, t)$, $t \in [0, T]$, as the limit
point in $C([0,T], \mathbb{R}^d)$, where $x \in \overline{B}_{1}(0)$. This is proven in 
Lemmas~\ref{difftozero} and \ref{closertoode}, the proofs of which
do not require $(BM1)(ii)$.
\subsection{Further weakening of (BM1) as a consequence of Theorem~\ref{AccuMain}}\label{furtherBM}
\paragraph{}
In this section we use the second stability theorem (Theorem~\ref{AccuMain})
to answer the following question: If
$\underset{c \to \infty}{\lim} h_c (x)$ does not exist for all $x \in \mathbb{R}^{d}$,
then what are the sufficient conditions for the stability and convergence of the algorithm?
\paragraph{}
Taking our cue from assumption $(A6)$, we replace
$(BM1)$ with the following assumption, call it $(BM4)$.
\begin{itemize}
 \item[(BM4)(i)] \textit{The function $h:\mathbb{R}^d\rightarrow
\mathbb{R}^d$
is Lipschitz continuous, with Lipschitz constant $L$. 
Define the set-valued map, $h_\infty(x) := 
\overline{co}\left( Limsup_{c \to \infty} \{h_c (x)\} \right)$, where
$x\in \mathbb{R}^d$.}
\\
Note that 
$Limsup_{c \to \infty} \{h_c (x)\} = \{y \ |\ 
\underset{c \to \infty}{\underline{\lim}} \ \lVert h_c(x) - y \rVert = 0\} $.
\item[(BM4)(ii)] \textit{$\dot {x}(t) \in h_{\infty}(x(t))$ has an 
 attracting set, $\mathcal{A}$, with $\overline{B}_{1}(0)$ as a subset of its
 fundamental neighborhood. This attracting set is such that $\mathcal{A} \subseteq B_{1}(0)$.}
\end{itemize}
\paragraph{}
Observe that $Limsup_{c \to \infty} \{h_c (x)\} =$ 
$\underset{c \to \infty}{\lim} h_c (x)$ when 
$\underset{c \to \infty}{\lim} h_c (x)$ exists.
Recall the definition of $Limsup$, the upper-limit of a sequence of sets, from Section~\ref{definitions}.
It can be shown that if a recursion given by (\ref{eq:basicrecursion}) satisfies assumptions
$(BM1)(i)$ and $(BM1)(iii)$ then it also satisfies $(BM4)(i),(ii)$. Assumption $(BM4)$
unifies the two possible cases: when the limit of $h_c$, as $c \to \infty$, 
exists for each $x \in \mathbb{R}^d$ and when it does not.
\paragraph{}
We claim that a recursion given by (\ref{eq:basicrecursion}), satisfying assumptions
$(BM2),\ (BM3)$ and $(BM4)$ will also satisfy $(A1)-(A3)$, $(A6)$ and $(A5)$ 
(see section~\ref{AccuMainSec}). From Theorem~\ref{AccuMain} it follows that
the iterates are stable and converge to a closed, connected, internally chain transitive
and invariant set of $\dot{x}(t) = h(x(t))$. The following generalization
of the \textit{Borkar-Meyn Theorem} is a direct consequence of Theorem~\ref{AccuMain}.
\begin{corollary}[Generalized Borkar-Meyn Theorem]\label{GenBorkarMeyn}
 Under assumptions $(BM2)$, $(BM3)$ and $(BM4)$,
almost surely the sequence $\{x_{n}\}_{n \ge 0}$ generated by the
stochastic recursive equation (\ref{eq:basicrecursion}), is bounded
and converges to a closed, connected, internally chain transitive and invariant
set of $\dot{x}(t) = h(x(t))$.
\end{corollary}
\begin{proof}
Assumptions $(A1)-(A3)$ and $(A6)$ follow directly from $(BM2)$, $(BM3)$ and
$(BM4)$. We show that $(A5)$ is also satisfied.
 Let $c_n \uparrow \infty$,
 $x_{n} \to x$, $y_n \to y$ and $y_n \in h_{c_{n}}(x_n)$ (here $y_n = h_{c_{n}}(x_n)$),
 $\forall \ n \ge 1$.
 It can be shown that 
 $\lVert h_{c_{n}}(x_{n}) - h_{c_{n}}(x) \rVert \le L \lVert x_n - x \rVert$.
 Hence we get that $h_{c_{n}}(x) \to y$. In other words,
 $\underset{c \to \infty}{\underline{\lim}} \lVert h_c (x) - y \rVert = 0$.
 Hence we have $y \in h_{\infty}(x)$. The claim now follows from Theorem~\ref{AccuMain}.
 \end{proof}
\section{Applications: The problem of approximate drifts \& stochastic gradient descent} \label{applications} 
\subsection{The problem of approximate drifts}\label{ApproximateDrift}
\paragraph{}
Let us recall the standard $SRE$:
\begin{equation}
 \label{eq:ODEIterate}
x_{n+1} = x_{n} + a(n) \left( h(x_{n}) + M_{n+1} \right),
\end{equation}
where $h: \mathbb{R}^{d} \longrightarrow \mathbb{R}^{d}$ is Lipschitz continuous, 
$\{a(n)\}_{n \ge 0}$ is the step-size sequence and $\{M_{n}\}_{n \ge 1}$
is the noise sequence.
\paragraph{}
The function $h$ is colloquially referred to as the drift. In many applications
the drift function cannot be calculated accurately.
This is referred to as the approximate drift problem.
For more details the reader is referred to \textit{Chapter 5.3} of \textit{Borkar} \cite{BorkarBook}.
Suppose the room for error is at most $\epsilon (> 0)$ then such 
an algorithm can be characterized by the following stochastic recursive inclusion:
\begin{equation}
 \label{eq:DIIterate}
 x_{n+1} = x_{n} + a(n) \left( y_{n} + M_{n+1} \right),
\end{equation}
 where $y_{n} \in h(x_n)  + \overline{B}_{\epsilon}(0)$ is an estimate of $h(x_{n})$ and
$\overline{B}_{\epsilon}(0)$ is the closed ball of radius $\epsilon$
around the origin. We define a new set-valued map called the
\textit{approximate drift} by $H(x) := h(x) + \overline{B}_{\epsilon}(0)$ for each $x \in 
\mathbb{R}^{d}$. In the following discussion we assume that $\epsilon \ge 0$. When
$\epsilon = 0$, the approximate drift algorithm described by (\ref{eq:DIIterate})
is really the $SRE$ given by (\ref{eq:ODEIterate}).
\paragraph{}
In this section we show the following: If
($\ref{eq:ODEIterate}$) satisfies $(BM2), (BM3)$ and $(BM4)$ 
then the corresponding approximate drift version given by ($\ref{eq:DIIterate}$)
satisfies $(A1)-(A5)$. 
For details on $(BM2)$ and $(BM3)$ see Section~\ref{BMRelax}; see Section~\ref{furtherBM}
for $(BM4)$.
We then invoke Theorem~\ref{AccuMain} to conclude that the iterates
converge to a closed, connected, internally chain transitive and invariant set associated
with $\dot{x}(t) \in h(x(t)) + \overline{B}_{\epsilon}(0) ( = H(x(t)))$. 
\paragraph{}
\textit{\textbf{For the remainder
of this section it is assumed that 
($\ref{eq:ODEIterate}$) satisfies $(BM2), (BM3)$ and $(BM4)$.}}
\begin{proposition}\label{pa3}
 $H(x) = h(x)+ \overline{B}_{\epsilon}(0)$ is a Marchaud map. Further, recursion
 (\ref{eq:DIIterate}) satisfies $(A1), \ (A2)$ and $(A3)$.
\end{proposition}
\begin{proof}
Since $\overline{B}_{\epsilon}(0)$ is convex and compact, it follows that
$H(x)$ is convex and compact for each $x \in \mathbb{R}^d$. 
Fix $x \in \mathbb{R}^{d}$ and $y \in H(x)$,
then $\lVert y \rVert \le \lVert h(x) \rVert + \epsilon$ and
$\lVert y \rVert \le \lVert h(0) \rVert$ + $L \lVert x - 0 \rVert$ $+ \ \epsilon$
since $h$ is Lipschitz continuous with Lipschitz constant $L$. If we set 
$K := \left( \lVert h(0) \rVert + \epsilon \right) \vee L$, then we get
$\lVert y \rVert \le K \left(1+ \lVert x \rVert \right)$.
This shows that $H$ is point-wise bounded.
To show the upper semi-continuity of $H$ assume $\underset{n \to \infty}{\lim}$
$x_{n} = x$,  $\underset{n \to \infty}{\lim}$ $y_n = y$ and $y_n \in H(x_n)$
for each $n \ge 1$. For all $n \ge 1$, $y_n  =  h(x_n) + z_n$ for some
$z_{n} \in \overline{B}_{\epsilon}(0)$. Further, $h(x_n) \to h(x)$ as $x_n \to x$.
Since both $\{y_n\}_{n \ge 1}$
and $\{h(x_n)\}_{n \ge 1}$ are convergent sequences, $\{z_n\}_{n \ge 1}$
is also convergent. Let $z :=$ $\underset{n \to \infty}{\lim}\ z_n$; 
$z$ is such that $z \in \overline{B}_{\epsilon}(0)$  
since $\overline{B}_{\epsilon}(0)$ is compact. Taking limits on both sides of
$y_n  =  h(x_n) + z_n$, we get $y  =  h(x) + z$. Thus $y \in H(x)$.
\paragraph{}
Since (\ref{eq:DIIterate}) is assumed to satisfy $(BM2)$ and $(BM3)$ 
it trivially follows that it satisfies $(A2)$ and $(A3)$.
\end{proof}
\paragraph{}
Before showing that (\ref{eq:DIIterate}) satisfies $(A4)$, we construct the 
following family of set-valued maps:
\begin{equation} \label{Hc}
 H_{c}(x) := \left\{ \frac{h(cx)}{c} \ + \ \frac{y}{c} \ | \ y \ \in \ \overline{B}_{\epsilon}(0) \right\},
\end{equation}
where $c \ge 1$ and $x \in \mathbb{R}^d$.
In other words, $H_{c}(x) = h_{c}(x)+ \overline{B}_{\epsilon / c}(0)$ for each
$x \in \mathbb{R}^{d}$. 
\begin{proposition}\label{pa4}
 (\ref{eq:DIIterate}) satisfies (A6).
\end{proposition}
\begin{proof}
To prove this it is enough to show that 
$H_\infty(x) = h_\infty(x)$, where $H_\infty (x) :=Limsup_{c \to \infty} H_{c}(x)$
and $h_\infty(x) :=Limsup_{c \to \infty} h_{c}(x)$.
Since $\dot{x}(t) \in h_\infty (x(t))$ satisfies $(BM4)(ii)$ it trivially follows that $(A6)$
is satisfied by (\ref{eq:DIIterate}). Note that $(BM4)(ii)$ and $(A6)$ essentially say the
same thing.
\paragraph{}
First we show $h_\infty (x) \subseteq H_\infty (x)$ for every $x \in \mathbb{R}^d$.
Let $y \in h_\infty(x)$, $\exists c_n \uparrow \infty$ such that $h_{c_n} \to y$ as
$c_n \uparrow \infty$. Since $h_{c_n}(x) \in H_{c_n}(x)$ it follows 
from the definition of $Limsup$ that $y \in H_\infty (x)$. To show $H_\infty(x) \subseteq h_\infty(x)$
we start by assuming the negation \textit{i.e.,} 
for some $x \in \mathbb{R}^d$ $\exists y \in H_\infty (x)$ such that 
$y \notin h_\infty(x)$. Let $c_n \uparrow \infty$ and $y_n \in H_{c_n}(x_n)$ such that 
$\underset{c_n \uparrow \infty}{\lim} y_n = y$. Since $\lVert y_n - h_{c_n}(x_n) \rVert \le 
\frac{\epsilon}{c_n}$ we have $\underset{c_n \uparrow \infty}{\lim} h_{c_n}(x_n) = y$.
We have the following:
\[
 \lVert y - h_{c_n}(x) \rVert \le \lVert y - h_{c_n}(x_n) \rVert +
 \lVert h_{c_n}(x_n) - h_{c_n}(x) \rVert .
\]
Taking limits on both sides we get that $\lVert y - h_{c_n}(x) \rVert \to 0$ \textit{i.e.,}
$y \in h_\infty(x)$. This is a contradiction.
  \end{proof}
  \paragraph{}
\begin{proposition} \label{pa5}
 (\ref{eq:DIIterate}) satisfies (A5).
\end{proposition}
 \begin{proof}
  Given $c_n \uparrow \infty$, $x_n \to x$, $y_n \to y$ and $y_n \in H_{c_n}(x_n)$ 
  $\forall n$, we need to show that $y \in H_{\infty}(x)$. As in the proof
  of Proposition~\ref{pa4} we have $\underset{c_n \uparrow \infty}{\lim} h_{c_n}(x_n) = y$.
  Since $\lVert h_{c_n}(x_n) - h_{c_n}(x) \rVert \le L \lVert x_n - x \rVert$ we have
  that $\underset{c_n \uparrow \infty}{\lim} \lVert h_{c_n}(x_n) - h_{c_n}(x) \rVert = 0$
  and $\underset{c_n \uparrow \infty}{\lim} h_{c_n}(x) = y$. In other words,
  $y \in h_\infty(x)$. In Proposition~\ref{pa3} we have shown that $h_\infty \equiv H_\infty$
  therefore $y \in H_\infty(x)$. 
   \end{proof}
\paragraph{}
\begin{corollary} \label{bmcorr}
 If a \textit{SRE}, given by (\ref{eq:ODEIterate}), satisfies $(BM2), (BM3)$ and $(BM4)(i),(ii)$ 
 then  the corresponding approximate drift version given by (\ref{eq:DIIterate}) is stable
 almost surely. In addition, it converges to a closed, connected, invariant and
 internally chain transitive set of $\dot{x}(t) \in H(x(t))$, where $H(x) = h(x) + 
 \overline{B}_\epsilon(0)$.
\end{corollary}
\begin{proof}
In Propositions~\ref{pa3}, \ref{pa4} and \ref{pa5} we have shown that (\ref{eq:ODEIterate})
 satisfies $(A1) - (A3), \ (A5), \ (A6)$;
 the statement now follows directly from Theorem~\ref{AccuMain}.   
 \end{proof}
\paragraph{}
\begin{remark}
In the context of (\ref{eq:ODEIterate}), we have that $h$ is Lipschitz and
$h_c:x \mapsto \frac{h(cx)}{c}$. Supposing $\underset{c \to \infty}{\lim} h_c (x)$
exists for every $x \in \mathbb{R}^d$ (see $(BM1)(i)$ in Section~\ref{GenBorkarMeynSec}) then
$\underset{c \to \infty}{\lim} h_c (x) =$
$Limsup_{c \to \infty} \{h_c(x)\}$.
Further, $Limsup_{c \to \infty} \{h_c(x)\}$ is non-empty for every $x \in \mathbb{R}^d$
(since $h_c (x) \le K(1 + \lVert x \rVert)$, $c \ge 1$), even if 
$\underset{c \to \infty}{\lim} h_c (x)$ does not exist for some $x \in \mathbb{R}^d$. 
Hence the analysis of the approximate drift problem in
this section is all encompassing. The aforementioned
is also the reason why in Section~\ref{furtherBM} we define 
$h_{\infty}(x) := \overline{co} (Limsup_{c \to \infty} \{h_c(x)\})$. It may be noted 
that we use
$Limsup_{c \to \infty} \{h_c(x)\}$ instead of $Limsup_{c \to \infty} h_c(x)$ since $Limsup$
acts on sets and $h$ (in this context) is a function that is not set-valued.
Finally, in Corollary~\ref{bmcorr} if we let
$\epsilon = 0$ then we may derive Corollary~\ref{GenBorkarMeyn}.
\end{remark}
\newpage
\subsection{Stochastic gradient descent}\label{sgdsec}
\paragraph{}
Stochastic gradient descent is a gradient descent optimization technique to find
the minimum set of a (continuously) differentiable function. Suppose we want to
find the minimum of $F: \mathbb{R}^d \to \mathbb{R}$ for which we can run
the following $SRE$:
\begin{equation} \label{sgd}
 x_{n+1}  = x_n - a(n) [\nabla F(x_n) + M_{n+1}],
\end{equation}
where $\nabla F : \mathbb{R}^d \to \mathbb{R}^d$ is upper-semicontinuous and
$\lVert \nabla F(x) \rVert \le K(1 + \lVert x \rVert)$ $\forall x \in \mathbb{R}^d$
(point-wise bounded).
$\{a(n)\}_{n \ge 0}$ is the  given step size sequence and $\{M_{n+1}\}_{n \ge 0}$
is the martingale difference noise sequence. If the assumptions
of Bena\"{i}m, Hofbauer and Sorin \cite{Benaim05} are satisfied by (\ref{sgd})
then the iterates converge to a closed, connected, 
internally chain transitive and invariant set of $\dot{x}(t) = - \nabla F(x(t))$
which is also the minimum set of $F$. In this section we shall not
distinguish between the asymptotic attracting set of $\dot{x}(t) = - \nabla F(x(t))$
and the minimum set of $F$.
\paragraph{}
As explained in Section~\ref{introduction}, while implementing (\ref{sgd}) one can only hope to calculate an approximate value
of the gradient at each step. However, one has control over the ``approximation error''.
This is typical when gradient estimators with fixed perturbation parameters are used, 
it could also be a consequence of the inherent computational capability of the 
computer used to run the algorithm. In reality one is running the following $SRI$:
\begin{equation} \label{sgdi}
 x_{n+1} = x_n + a(n) [y_n + M_{n+1}],
\end{equation}
where $y_n \in - \nabla F(x_n) + \overline{B}_ \epsilon (0)$ and $\epsilon > 0$
is the ``approximation error''. The following questions are natural:
\begin{enumerate}
 \item Are the iterates stable?
 \item If so, where do they converge?
\end{enumerate}
Define the following set valued map, $H: x \mapsto - \nabla F(x) + \overline{B}_ \epsilon (0)$.
 As in (\ref{Hc}) we define $H_c (x) :=  \frac{- \nabla F(cx)}{c} + \overline{B}_{\epsilon / c} (0)$
 and $H_\infty (x) := Limsup_{c \to \infty} H_c (x)$ = $Limsup_{c \to \infty} 
 \left\{ \frac{- \nabla F(cx)}{c} \right\}$. Recall the definition of $Limsup$
 from Section~\ref{definitions}.
 \begin{proposition}
  (\ref{sgdi}) satisfies $(A1)$  \textit{i.e.,} $H$ is a marchaud map.
 \end{proposition}
 \begin{proof} 
  Given $x_n \to x$, $y_n \to y$ and $y_n \in H(x_n) \ \forall n$, we need to show that
  $y \in H(x)$. For each $n$ we have $y_n = - \nabla F(x_n) + z_n$,
  where $z_n \in \overline{B}_\epsilon(0)$.
  Since $\nabla F$ is point-wise bounded, it follows that
  $\{- \nabla F(x_n)\}$ is a bounded sequence. Let $\{n(m)\} \subseteq \mathbb{N}$ such that
  $\nabla F(x_{n(m)}) \to \nabla F(x)$, $y_{n(m)} \to y$. The subsequence $ z_{n(m)} \to z$
  for some $z \in \overline{B}_\epsilon (0)$ \textit{i.e.,}
  \[
  \left( -\nabla F(x_{n(m)}) + z_{n(m)} \right) \to \left(-\nabla F(x) + z \right) \in H(x).
  \]
   \end{proof}
   \paragraph{}
If in addition to $(A1)$, equation (\ref{sgdi}) also satisfies $(A2), (A3), (A5)$ and
$(A6)$ then it follows from Theorem~\ref{AccuMain} that the iterates are stable and
converge to a closed, connected, 
internally chain transitive and invariant set of $\dot{x}(t) \in \left( 
- \nabla F(x(t)) + \overline{B}_\epsilon (0) \right)$. 
\paragraph{}
\textbf{\textit{Suppose $F$ has the quadratic form $x^T A x + Bx + c$, where $A$ is a positive definite matrix, 
$B$ is some matrix and $c$ is some vector. Then it can be shown that
$(A1), \ (A2), (A3), (A5)$ and $(A6)$ are satisfied by (\ref{sgdi}) and the iterates are
stable and converge to a closed, connected, 
internally chain transitive and invariant set of $\dot{x}(t) \in 
- (A x(t) + B) + \overline{B}_\epsilon (0) $. If the comments in Remark \ref{a4'} are
incorporated \textit{i.e.,} we use $(A6)'$ instead of $(A6)$
then matrix $A$ need not be positive definite anymore.}}
\paragraph{}
For the purpose of this discussion assume that $\nabla F$ is Lipschitz continuous.
The graph of a set-valued map $H:\mathbb{R}^d \to \{ \text{subsets of }\mathbb{R}^d \}$ 
is given by $Graph(H) = \{(x,y) \ | \ x \in \mathbb{R}^d, \ y \in H(x)\}$. It is easy
to see that $Graph(- \nabla F + \overline{B}_\epsilon (0)) \subseteq
N^{2 \epsilon} \left( Graph(-\nabla F) \right)$.
Let us also assume that $\mathcal{A}$ is the global attractor (minimum set of $F$) of 
$\dot{x}(t) = - \nabla F(x(t))$ then 
every compact subset of $\mathbb{R}^d$ is its fundamental neighborhood.
It follows from the stability of the iterates that they will remain within a compact subset,
say $\mathcal{U}$, that may be sample path dependent.
It follows from Theorem 2.1 of
Bena\"{i}m, Hofbauer and Sorin \cite{Benaim12} that for all $\delta > 0$
there exists $\epsilon > 0$ such that $\mathcal{A}^\delta \subseteq N^\delta (\mathcal{A})$ 
is the attractor set of $\dot{x}(t) \in - \nabla F(x(t)) + \overline{B}_\epsilon (0)$.
Further, the fundamental neighborhood of $\mathcal{A}^\delta$ is $\mathcal{U}$ itself.
In other words, suppose we want to ensure convergence of the iterates to a $\delta-neighborhood$ of
the minimum set $\mathcal{A}$ then the ``approximation error'' should be
at most $\epsilon$ ($\epsilon$ is dependent on $\delta$).
\section{Final discussion on the generality of our framework}\label{finaldisc}
\paragraph{}
As explained in Section~\ref{sec3}, we run a projective scheme to show stability.
In other words, time is divided into intervals of length $T$; the iterates are checked
at the
beginning of each time interval to see if they are outside the unit ball;
all the iterates corresponding to $[T_n , T_{n+1} )$ are scaled by $r(n) =\lVert x(T_n) \rVert \vee 1$
\textit{i.e.,} the iterates are projected onto the unit ball around the origin.
For $t(m(n)) = T_n \le t(m(n) +k) < T_{n+1}$ we have the following re-scaled iterate:
\[
\frac{\overline{x}(t(m(n) +k))}{r(n)} = \frac{\overline{x}(t(m(n)))}{r(n)} +
\sum_{j=0}^{k-1} a(m(n) + j) 
\left[ \frac{\overline{y}(t(m(n) +j))}{r(n)} + \frac{M_{m(n)+j+1}}{r(n)} \right].
\]
In the above, $\frac{\overline{y}(t(m(n) +j))}{r(n)} \in \ 
h_{r(n)} \left( \frac{\overline{x}(t(m(n) +j))}{r(n)} \right)$. 
Since we have to worry about $r(n)$ running off to infinity it is natural
to define $h_\infty(x)$ to include all accumulation points of $\{h_c(x) \mid c \ge 1, \ c \to \infty\}$.
This is precisely what the $Limsup$ function (see Section~\ref{definitions}) allows us to do.
In Lemma~\ref{closertoode} it was shown that
the scaled iterates track a solution to $\dot{x}(t) \in h_\infty (x(t))$ provided the original iterates
are unstable \textit{i.e.,} $\underset{n}{\sup}\ r(n) = \infty$. Assumptions
$(A4)/(A6)$ were never used up to this point. At this stage it seems natural to
impose restrictions on $\overline{x}(t) \in h_\infty (x(t))$ to elicit the stability
of the original iterates.
\paragraph{}
As explained in Section~\ref{AccuMainSec} $Limsup_{c \to \infty} h_c(x)$ is non-empty
for every $x \in \mathbb{R}^d$ since $h$ is point-wise bounded.
Further, $h_\infty \equiv \overline{co} \left( Limsup_{c \to \infty} h_c \right)$
is shown to be Marchaud and the $DI$ $\dot{x}(t) \in h_\infty (x(t))$ 
has at least one solution. Assumption $(A6)$ is the restriction referred to 
in the previous paragraph that is imposed
to elicit the stability of the original iterates. On a related note, 
if $Liminf_{c \to \infty} h_c$ were non-empty, then we define 
$h_\infty \equiv \overline{Liminf_{c \to \infty} h_c}$ and check if $(A4)$ is satisfied.
\paragraph{}
If the $DI$ $\dot{x}(t) \in h_\infty (x(t))$ has global attractor inside $B_1(0)$,
then this is a sufficient condition for $(A6)$ to hold,
it then follows from Theorem~\ref{AccuMain} that the original iterates are stable
and converge to a closed connected internally chain transitive set associated with
$\dot{x}(t) \in h(x(t))$.
More generally, in lieu of
Remark~\ref{a4'} it is sufficient that the $DI$ has some global attractor, not
necessarily inside the unit ball, since $(A6)'$ will then hold. 
This in turn implies stability.
\paragraph{}
In case of the original Borkar-Meyn assumptions,
$(BM1)(i),(ii)$ (see Section~\ref{GenBorkarMeynSec}) 
needed to be checked even before we could define $h_\infty$ while
in our case we do not need any extra assumptions to define $h_\infty$. As explained
before, constructing a global Lyapunov function for $h_\infty$ is one of many sufficient
conditions that guarantee $(A4)'/(A6)'$. In case of Lyapunov-type conditions for
stability, additional properties of the constructed global Lyapunov function need to
be verified before we get stability, see \cite{Andrieu} for more details. 
However, to the best of our knowledge, there are no Lyapunov-type conditions
that guarantee stability of stochastic approximation algorithms with set-valued
mean fields ($SRI$), the class of algorithms dealt with in this paper. 
Hence our assumptions are
general and relatively easy to verify.
 \section{Conclusions}
 \paragraph{}
 An extension was presented to the theorem of \textit{Borkar and Meyn} to
 include approximation algorithms with set-valued mean fields. Two different sets of
 sufficient conditions were presented that guarantee the `stability and convergence' 
 of stochastic
 recursive inclusions. 
 As a consequence of Theorems \ref{main} \& \ref{AccuMain}, the original
 Borkar-Meyn theorem is shown to hold under weaker requirements.
 Further, as a consequence of Theorem~\ref{AccuMain}, we obtained a solution to the 
 ``approximate drift'' problem. 
 Prior to this paper, there was no proof of stability of 
 stochastic gradient descent algorithms that use constant-error gradient estimators.
 Hence we could only conclude that the iterates converge to a small neighborhood, say $\overline{N}$, of
 the minimum set with very high probability. In Section~\ref{sgdsec}
 we used our framework to show the stability of the aforementioned algorithm which in turn
 allowed us to conclude an almost sure convergence to $\overline{N}$.
 \paragraph{}
An important future direction would be to extend these results to the case when 
the set-valued drift is governed by a Markov process in addition to the iterate sequence. For
the case of stochastic approximations, such a situation has been considered in
[\cite{BorkarBook}, Chapter 6], where the Markov `noise' is tackled using the 
`natural timescale averaging' properties of stochastic approximation. Finally, 
it would be interesting
to develop Lyapunov-type assumptions for stability of stochastic algorithms with set-valued mean fields.
\section{Acknowledgments}
This work was partly supported by Robert Bosch Centre for Cyber-Physical Systems, IISc.
\bibliographystyle{plain} 
\bibliography{MOR-template} 

\end{document}